\documentclass[a4paper,journal]{IEEEtran}

\usepackage{cite}
\usepackage{algorithm}
\usepackage{algorithmicx}
\usepackage[noend]{algpseudocode}
\usepackage{balance}
\usepackage{graphicx}
\usepackage{tabularx}
\usepackage{epstopdf}
\usepackage{url}
\usepackage{amssymb}
\usepackage{amsthm}
\usepackage{tikz}
\usepackage{textcomp}
\usepackage{marvosym}
\usepackage{subfig}
\usepackage{newtxmath}
\usepackage{enumerate}
\usepackage{mathtools,amsthm}

\newtheorem{theorem}{Theorem}
\newtheorem{lemma}{Lemma}
\newtheorem{corollary}{Corollary}
\newtheorem{assumption}{Assumption}

\newtheorem{remark}{Remark}

\DeclareMathOperator{\diag}{diag}
\DeclareMathOperator{\col}{col}
\DeclareMathOperator{\sgn}{sgn}

\DeclareMathOperator{\tr}{tr}
\newcommand*\mystrut[1]{\vrule width0pt height0pt depth#1\relax}

\begin{document}
\newcommand\sForAll[2]{ \ForAll{#1}#2\EndFor} 
\newcommand\sIf[2]{ \If{#1}#2\EndIf}          

\title{Distributed Adaptive Time-Varying Optimization with Global Asymptotic Convergence}

\author{Liangze Jiang, Zheng-Guang Wu, and Lei Wang 

\thanks{The authors are with College of Control Science and Engineering, Zhejiang University, Hangzhou 310027, China (e-mail: zetojiang; nashwzhg; lei.wangzju @zju.edu.cn). (\em Corresponding author: Zheng-Guang Wu.)}}

\maketitle

\begin{abstract}
In this note, we study distributed time-varying optimization for a multi-agent system. We first focus on a class of time-varying quadratic cost functions, and develop a new distributed algorithm that integrates an average estimator and an adaptive optimizer, with both bridged by a Dead Zone Algorithm. Based on a composite Lyapunov function and finite escape-time analysis, we prove the closed-loop global asymptotic convergence to the optimal solution under mild assumptions. Particularly, the introduction of the estimator relaxes the requirement for the Hessians of cost functions, and the integrated design eliminates the waiting time required in the relevant literature for estimating global parameter during algorithm implementation. We then extend this result to a more general class of time-varying cost functions. Two examples are used to verify the proposed designs.
\end{abstract}


\section{Introduction}\label{Notations}
Designing and reconstructing the dynamic behavior of groups through optimization and control techniques propels the advancement of distributed optimization theory \cite{Nedic09,wang2010control,gharesifard2013distributed}, which also penetrates into physical applications such as sensor networks \cite{Rabbat04}, robot networks \cite{bhattacharya2011distributed,QIN2022110113} and power systems \cite{Mahmoud16,huang2020distributionally}. Recently, there is a trend towards modeling and solving large-scale optimization problems in dynamic environments (e.g., see \cite{su2009traffic,Simonetto20,LI2023100904,HATANAKA2016210,Shorinwa24} and the references therein). To cope with constantly changing environmental information, the modeling of cost functions in dynamic environments is typically based on real-time information {\em acquisition} by agents (see Fig. \ref{fig.model})\cite{su2009traffic,Simonetto20,LI2023100904,HATANAKA2016210,Zhou22,Shorinwa24,Intelligent2017}. 
\begin{figure}[!h]
\centerline{\includegraphics[width=0.8\columnwidth]
    {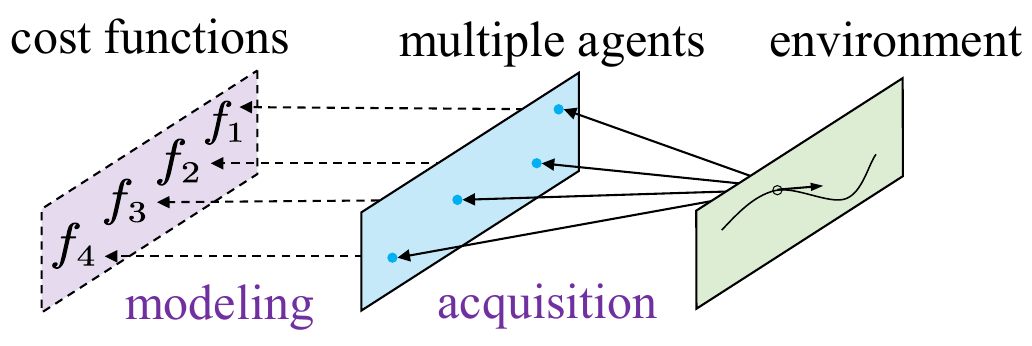}}
    \caption{Modeling of cost functions in dynamic environments.}
    \label{fig.model}
\end{figure}
However, the dynamic nature of the environment may disrupt this process, which in turn leads to the `fragility' of cost functions. For instance, in multi-camera localization of moving objects, cameras with fixed positions fail to capture image of mobile objects outside their field of view\cite{HATANAKA2016210}. In multi-robot target tracking, obstacles or limited perceptual distance prevent some robots from obtaining target information\cite{Zhou22,Shorinwa24}. When information acquisition is hindered due to environmental changes, weak or missing local cost functions cannot guarantee their uniform strong convexity about decision variables (e.g., see Example 2 in Section \ref{section.simulation}). Moreover, a realistic topic in autonomous optimization is to reduce the dependence of algorithms on global parameter of the cost functions and the communication graph\cite{ZOU2020}. These observations prompt us to investigate distributed time-varying optimization.

Intuitively, during the search for the time-varying optimal solution, gradient information provides only error feedback relative to the current optimal solution. By introducing the time derivative of the gradient as compensation\cite{su2009traffic,Simonetto16,Fazlyab18}, the modified gradient descent method can solve the time-varying optimization problem exactly\footnote{The considered time-varying optimization setting is different from the online optimization\cite{Simonetto20,LI2023100904} where the structural information on the cost functions is unavailable and non-zero tracking errors are pursued. In contrast, in our problem setting, each agent can construct a complete local cost function by adequately acquiring information about the environment.}. This basic idea is further expanded to distributed versions to tackle large-scale time-varying optimization problems in \cite{Ren17,Chu22,Wang22,Sun23,Ding22,Huang20}, some of which consider equality \cite{Chu22,Ding22,Wang22,Sun23} and inequality \cite{Sun23} constraints. In \cite{Sun17}, a robust control method is introduced to solve distributed time-varying quadratic optimization (DTQO) with coupled cost functions. However, as mentioned earlier, the cost functions established in dynamic environments may exhibit time-varying and non-strongly convex properties. In such cases, the existing algorithms cannot be directly applicable, as they require the Hessians of all local cost functions to be positive definite\cite{Sun17,Chu22,Ding22,Wang22,Huang20,Sun23,Santilli24,Zheng24}, identical\cite{Ren17,Huang20}, diagonal \cite{Wang22} and time-invariant\cite{Ding22,Santilli24}. Although one design in \cite{Ren17} considers nonpositive Hessians of the local cost functions, its assumption of constant upper bounds on certain complex functions limits the types of cost functions. Moreover, its convergence result is semi-global since the initial states of agents are needed for the design. Additionally, the waiting time in \cite{Ren17,Chu22,Wang22} implies a two-stage algorithm implementation, requiring knowledge of the initial states of agents, the communication graph, or other prior global parameters of the cost functions to determine the size of it before running the algorithm.
 
Motivated by this, we start with a basic DTQO and develop a distributed algorithm that seamlessly integrates an average estimator with an adaptive optimizer, based on the nonsingular matrices produced by the Dead Zone Algorithm (DZA). We then extend this result to address the distributed optimization problem of a general class of cost functions, not limited to quadratic cost functions\cite{Sun17,Wang22,Santilli24}. The contributions lie in: ({\romannumeral1}) We propose distributed optimization algorithms without relying on positive definite, identical, diagonal or time-invariant Hessians of all local cost functions. To the best of our knowledge, these are the first result of distributed optimization for the aforementioned types of cost functions. ({\romannumeral2}) By constructing a composite Lyapunov function and introducing the concept of the finite escape time, we establish the asymptotic convergence of closed-loop outputs to the optimal solution in the global sense, in contrast with the semi-global result in \cite{Ren17}. ({\romannumeral3}) Different from \cite{Ren17,Chu22,Wang22}, the proposed design does not require the waiting time during algorithm implementation. Such design also relaxes the assumption that the upper bounds of certain time-varying parameters or their high-order time derivatives are {\em known} in advance \cite{Ren17,Chu22,Sun17}. 

{\bf Notations.} The set of real numbers is represented by $\mathbb{R}$. Let $1_N$ denote the vector $(1,\cdots,1)^\top \in\mathbb{R}^N$ and $I_m$ denote the $m$-dimensional identity matrix. For a matrix $A\in\mathbb{R}^{m\times m}$, we use $\lambda_p(A)$, $p\in\{1,\cdots,m\}$, to represent the $p$-th smallest eigenvalue of $A$. For a vector or a matrix $X$, we use $|X|_p$ to represent the $l_p$-norm of $X$. $|X|$ is the abbreviation of $|X|_2$. The signum function is denoted by $\sgn(y)$, with $y\in\mathbb{R}$. For a vector $x={(x_1,\cdots,x_n)^\top}$, define $\sgn(x)={(\sgn(x_1),\cdots,\sgn(x_n))^\top}$ and $S(x)={(S(x_1),\cdots,S(x_n))^\top}$, where $S(x_i)=x_i/(|x_i|+\epsilon_1\eta_t)$ with $\epsilon_1>0$ and $\eta_t\in\mathbb{R}$. For a matrix $A=[a_{ij}]\in\mathbb{R}^{m\times n}$, define $\sgn(A)=[\sgn(a_{ij})]\in\mathbb{R}^{m\times n}$ and use $A_{\otimes}$ to denote $A\otimes I_m$. Denote $\diag\{A_i\}$, $i=1,\cdots,N$ as a diagonal matrix with diagonal entries composed of $A_1,\cdots,A_N$, where $A_i\in\mathbb{R}^{m_i\times m_i}$. The gradient of a function $f(x,t):\mathbb{R}^m\times\mathbb{R}_+\rightarrow\mathbb{R}$ with respect to $x$ is denoted by $\nabla f(x,t)$. Let $\nabla_t f(x,t)$ denote the partial derivative of ${\nabla} f(x,t)$ with respect to time $t$.

\section{Preliminaries and Problem Statement}
\label{Problem Formulation}
\subsection{Preliminaries}
A differentiable function $f(x,t):\mathbb{R}^m\times\mathbb{R}_+\rightarrow\mathbb{R}$ is uniformly $\zeta$-strongly convex in $x$ if there exists $\zeta>0$ such that $(a-b)^{\top}\left(\nabla f(a,t)-\nabla f(b,t)\right)\geq \zeta |a-b|^2$, $\forall a,b\in \mathbb{R}^m$, $\forall t\geq0$. $\nabla f(x,t)$ is uniformly $\theta$-Lipschitz in $x$ if there exists $\theta>0$ such that $|\nabla f(a,t)-\nabla f(b,t)|\leq \theta|a-b|$, $\forall a,b\in \mathbb{R}^m$, $\forall t\geq0$. 

The communication network is described by an undirected graph $\mathcal{G}=(\mathcal{N},\mathcal{E})$, where $\mathcal{N}=\left\{1,\dots,N\right\}$ is the node set and $\mathcal{E}\subseteq \mathcal{N}\times \mathcal{N}$ is the edge set. Denote the weighted adjacency matrix of $\mathcal{G}$ as $\mathcal{A}=[a_{ij}]_{N\times N}$. An edge $(i,j)\in \mathcal{E}$ means that nodes $i,j$ can receive information from each other. Define $a_{ij}=1$ if $(j,i)\in \mathcal{E}$, and $a_{ij}=0$ otherwise. Let $\mathcal{N}_i=\{j\in\mathcal{N}:(j,i)\in\mathcal{E}\}$ denote the set of node $i$'s neighbors. The Laplacian of a graph $\mathcal{G}$ is defined as $L=[l_{ij}]_{N\times N}$ with $l_{ii}=\sum_{j\neq i}a_{ij}$ and $l_{ij}=-a_{ij}$ for $j\neq i$. If there is a path from any node to any other node in graph $\mathcal{G}$, it is called connected. Denote the incidence matrix associated with $\mathcal{G} $ as $D=[d_{ik}]_{N\times |\mathcal{E}|}$, where $d_{ik}=1$ if the edge $e_k$ enters the node $i$, $d_{ik}=-1$ if the edge $e_k$ leaves the node $i$, and $d_{ik}=0$ otherwise.

\begin{assumption}\label{assumption.graph1}
(Graph connectivity) The graph $\mathcal{G}$ is undirected and connected.
\end{assumption}

By Assumption \ref{assumption.graph1}, it is known that the Laplacian matrix $L$ of the graph possesses a zero eigenvalue, corresponding to the eigenvector $1_N$, while the remaining eigenvalues are positive \cite{Algebraic}, which satisfy $\lambda_1(L) = 0 < \lambda_2(L) \leq\cdots \leq \lambda_N(L)$.

\subsection{Problem Definition}
Consider a group of $N$ agents interacting over a communication network $\mathcal{G}=(\mathcal{N},\mathcal{E})$. Each agent $i\in\mathcal{N}$ is assigned a local cost function defined by 
\begin{align}\label{eq.fi}
    f_i(x_i,t)=\frac{1}{2}x_i^{\top}H_i(t)x_i+R_i^{\top}(t)x_i+d_i(t),
\end{align}
where $x_i\in\mathbb{R}^m$ is the decision variable, and $H_i(t)\in\mathbb{R}^{m\times m}$, $R_i(t)\in\mathbb{R}^m$ and $d_i(t)\in\mathbb{R}$ are time-varying parameters. We aim to design a distributed algorithm such that each agent solves the optimization problem
\begin{align}\label{optimizationproblem}
&\min_{r\in\mathbb{R}^{m}}f(r,t)=\sum_ {i\in \mathcal{N}}{f_i(r,t)}.
\end{align}
The optimal solution of the above optimization problem is defined as $r^*(t)=\mathop{\arg\min}_{r\in\mathbb{R}^m}f(r,t)$. To guarantee the first-order optimality condition $\nabla{f}(r^*(t),t)=0$ holds, $\forall t\geq 0$, there is $(d/dt){\nabla} f(r^*(t),t)=H(t)\dot{r}^*(t)+\nabla_t f(r^*(t),t)\equiv0$, where $H(t)=\sum_{i\in \mathcal{N}}H_i(t)$. Then, if $H(t)$ is invertible, it follows that $\dot{r^*}(t)\equiv-H^{-1}(t)\nabla_t f(r^*(t),t)$. The following assumption guarantees the uniqueness of $r^*(t)$ and the positive definiteness of the Hessian $H(t)$, $\forall t\geq0$.

\begin{assumption}\label{assumption.lipschitz.stronglyconvex}
(Uniform strong convexity) The global cost function $f(r,t)$ is uniformly $h_1$-strongly convex with respect to $r$. That is, there exists a known $h_1>0$ such that ${\lambda_1}(H(t))\geq h_1$, $\forall t\geq0$.
\end{assumption}

Compared with most existing literature \cite{Sun17,Chu22,Ding22,Wang22,Huang20,Sun23,Santilli24,Zheng24}, each $f_i(x_i,t)$ here is not required to be uniformly strongly convex. It is even allowed to be {\em non-convex}. Moreover, the Hessians of all $f_i(x_i,t)$, $i\in\mathcal{N}$, do not have to be diagonal \cite{Wang22} or time-invariant \cite{Ding22}, nor do they have to be completely identical as required in \cite{Ren17} and \cite{Huang20}.

\begin{remark}
If the local cost functions are all convex, and there exists a strongly convex one with a known strongly convex coefficient (or a more conservative lower bound) $\underline{h}$, then, $h_1$ can directly take $\underline{h}$ (see \cite[Corollary 4.3.15]{matrix} for details).
\end{remark}

\begin{assumption}\label{assumption.boundedness.hessian}
(Boundedness) For any $i\in\mathcal{N}$, $|{H}_i(t)|$, $|{R}_i(t)|$ and $|\dot{R}_i(t)|$ are bounded, $\forall t\geq0$, but their upper bounds are unknown. Moreover, there exists a known ${h}_2>0$ such that $|\dot{H}_i(t)|_\infty\leq{h}_2$, $\forall t\geq0$.
\end{assumption}
We stress that Assumption \ref{assumption.boundedness.hessian} guarantees the linear boundedness of $\nabla_t f_i(x_i,t)$ with respect to $x_i$, which is more relaxed than the assumptions that directly impose {\em known} and {\em constant} bounds on $\nabla_t f_i(x_i,t)$ in \cite{Ren17} and $(d/dt)\nabla_t f_i(x_i,t)$ in \cite{Chu22}. By the derivations below \eqref{optimizationproblem}, Assumptions \ref{assumption.lipschitz.stronglyconvex} and \ref{assumption.boundedness.hessian} guarantee the boundedness of $r^*(t)$ and thus $\dot{r}^*(t)$. Moreover, we only require knowledge of the lower bound of the Hessian's eigenvalue and the upper bound of Hessian's changing rate. Additional prior knowledge of cost functions, such as the upper bounds of $(\partial/\partial t)\nabla_t f_i(x_i,t)$ in \cite{Ren17,Chu22}, $\dddot{r}^*(t)$ in \cite{Sun17}, $\ddot{R}_i(t)$ in \cite{Ren17,Chu22} and $\nabla f_i(x_i,t)$ in \cite{Sun23}, is not necessary here.

\begin{remark}\label{remarkaa}
In practice, the parameters of cost functions are typically constrained by the properties of physical objects. For example, consider the problem of multi-UAV tracking of a target\cite{Zhou22,Shorinwa24} with the cost function $f_i=|x_i-r(t)|^2$, where $r(t)$ is the trajectory of the target vehicle. The limited energy and engine power of the vehicle, in fact, must constrain $|r(t)|$ and $|\dot{r}(t)|$ within a certain range, respectively, but the specific upper bounds of $|r(t)|$ and $|\dot{r}(t)|$ may be unknown in advance. Assumption \ref{assumption.boundedness.hessian} aligns with these observations. 
\end{remark}

\section{Design Based on Estimator and Optimizer}\label{framework}

Revisiting the centralized optimization algorithm in \cite{su2009traffic}, we can design $\dot{r}=-\nabla f(r,t)-H^{-1}(t)\nabla_t f(r,t)$ to guide $r(t)$ to asymptotically converge to the optimal $r^*(t)$. Inspired by this, the following design adopts distributed `$\nabla f_i$' feedback and `$NH^{-1}(t)\nabla_t f_i$' feedforward\footnote{The scaling factor $N$ is instrumental for the subsequent derivations.}. In view of the heterogeneity of $H_i(t)$, $i\in\mathcal{N}$, the technique of dynamic average consensus is employed to estimate the average of all Hessians, i.e., $H(t)/N$.

\subsection{Average Estimator and Output Reconstruction}
\label{DAE}
For agent $i\in\mathcal{N}$, we first consider the average estimator
\begin{align}\label{eq.z}
    \dot{\xi}_i=&-\omega\sum_{j\in\mathcal{N}_i}\sgn(z_i-z_j),~~~~z_i=\xi_i+H_i(t),
\end{align}
where $\omega>{h}_2$, $\xi_i\in \mathbb{R}^{m\times m}$, and $z_i\in \mathbb{R}^{m\times m}$ is the output of estimator. The initial conditions satisfy $\sum_{j\in\mathcal{N}}\xi_i(0)=0$ and $\xi_i(0)$ is symmetric. The convergence of $z_i$ to $\bar{H}(t):=H(t)/N$ in \emph{finite time} can be guaranteed by \cite{fei12}. However, since the $H_i(t)$ may not be positive definite, $z_i(t)$ inevitably passes through singular points during its evolution. Considering $z^{-1}_i(t)$ as a part of feedforward term, we propose Algorithm \ref{alg} to avoid the impact of singular $z_i(t)$, where $h_0>0$ is to be determined and $\mathcal{M}=\{1,\cdots,m\}$. In Algorithm \ref{alg}, we use ${\min}_{p\in\mathcal{M}}|\lambda_p(z_i(t))|$ to quantify how close $z_i(t)$ approaches singular points, and constrain $\hat{z}_i(t)$ to be within a `safe range' away from singular points. Consequently, ${\min}_{p\in\mathcal{M}}|\lambda_p(\hat{z}_i(t))|\geq h_0$ holds for all $t\geq0$, regardless of whether the input matrix $z_i(t)$ is singular. We thus refer to Algorithm \ref{alg} as the {\bf D}ead {\bf Z}one {\bf A}lgorithm, and the inverse of the output of the DZA, ${\hat{z}_i}^{-1}(t)$, remains well-defined during the evolution of $z_i(t)$. 

\begin{algorithm}[!h]
    \caption{Dead Zone Algorithm of $z_i(t)$} \label{alg} 
    \textbf{Input:} $z_i(t)$~~~\textbf{Output:} $\hat{z}_i(t)$
    \begin{algorithmic}[1]
    \State Initialize $H_0\leftarrow h_0 I_m$
    \While{$t\geq0$}
        \State $\lambda_0\leftarrow\mathop{\min}_{p\in\mathcal{M}}|\lambda_p(z_i(t))|$
        \If {$t=0$} 
            \If {$\lambda_0<h_0$}
                \State $\hat{z}_i(t) \leftarrow H_0$
            \Else
                \State $\hat{z}_i(t) \leftarrow z_i(t)$
            \EndIf
        \Else
            \If {$\lambda_0<h_0$}
                \State $\hat{z}_i(t) \leftarrow H_0$
            \Else
                \State $\hat{z}_i(t) \leftarrow z_i(t)$,~~~$H_0\leftarrow z_i(t)$
            \EndIf
        \EndIf
    \EndWhile
    \end{algorithmic}
    \end{algorithm}

From the perspective of ensuring nonsingularity, the DZA plays the same role as the projection method in \cite{Ren17}. However, since the integrated design shown in Fig. \ref{fig.control} is adopted in our work, we also need to provide a sufficient condition for parameter selection in the DZA to ensure closed-loop convergence, rather than just ensuring nonsingularity. The following lemma illustrates an important relation between $z_i(t)$ and the `shaped' $\hat{z}_i(t)$, if we select a suitable parameter $h_0$. This relation is instrumental for the analysis in Section \ref{ConvergenceAnalysis}.
\begin{lemma}\label{condition}
    Suppose Assumptions \ref{assumption.graph1}-\ref{assumption.boundedness.hessian} hold. For system \eqref{eq.z} and Algorithm \ref{alg} with $h_0=\gamma h_1/N$, there exist a constant ${k}_1>0$ such that $|\hat{z}_i(t)-\bar{H}(t)|\leq{k}_1|z_i(t)-\bar{H}(t)|$ holds $\forall t\geq 0$.
\end{lemma}
\begin{proof}
See Appendix \ref{appendix.proof2}.
\end{proof}

\begin{remark}
Different from \cite{Ren17,Chu22,Wang22,Sun23}, although the finite-time average estimator is adopted here, the size of convergence time (i.e., the waiting time in \cite{Ren17,Chu22,Wang22,Sun23}) is unnecessary for agents to obtain, because we do not employ the execution logic of a two-stage algorithm implementation.
\end{remark}

\subsection{Adaptive Optimizer}\label{parta}
For agent $i\in\mathcal{N}$, consider the following adaptive optimizer
\begin{align}
\dot{x}_i=&\underbrace{{\mystrut{2.5ex}}\smash {-\sum_{j\in\mathcal{N}_i}(\alpha_{ij}(x_i-x_j)+\beta_{ij}S(x_i-x_j))}}_{\text{\small Adaptive consensus coordinator}}\hspace{1.2mm}+\phi_i,\label{eq.ui}\\
\phi_i=&-k\nabla f_i(x_i,t)-\hat{z}_i^{-1}(t)\nabla_t f_i(x_i,t)\label{eq.phi},\\
\dot{\alpha}_{ij}=&|x_i-x_j|^2,~~~\dot{\beta}_{ij}=|x_i-x_j|_1-m\epsilon_1\eta_t,~~j\in\mathcal{N}_i,\label{eq.alpha}
\end{align}
where the function $S(\cdot)$ is defined in {\bf Notations}, $\eta_t=\exp(-\epsilon_2t)$ with $\epsilon_2>0$, $k>0$ is to be determined, $\alpha_{ij}(0)=\alpha_{ji}(0)\geq0$ and $\beta_{ij}(0)=\beta_{ji}(0)\geq m\epsilon_1/\epsilon_2$ with $\epsilon_1>0$. As a continuous approximation of the signum function, the function $S(\cdot)$ reduces the chattering of $x_i$ in application.

In \eqref{eq.ui}, $\phi_i$ incorporates distributed `$\nabla f_i$' feedback and `$\hat{z}_i^{-1}(t)\nabla_t f_i$' feedforward. The consensus coordinator utilizes adaptive control gains $\alpha_{ij}$ and $\beta_{ij}$ to handle unknown global parameter (see details in Appendix \ref{appendix.proof3}). Fig. \ref{fig.control} illustrates the structure of the proposed design for agent $i$, where dashed arrows represent information received from neighboring agents. Note that if $x_i$ denotes the position of agent $i$, and $x_i-x_j$ can be obtained through relative measurements (e.g. the acoustic and vision sensors), only the variable $z_j$ is transmitted through the communication channel. Therefore, this design reduces the communication resources usage compared to \cite{Ren17,Chu22,Wang22}.
\begin{figure}[!h]
\centerline{\includegraphics[width=0.9\columnwidth]
    {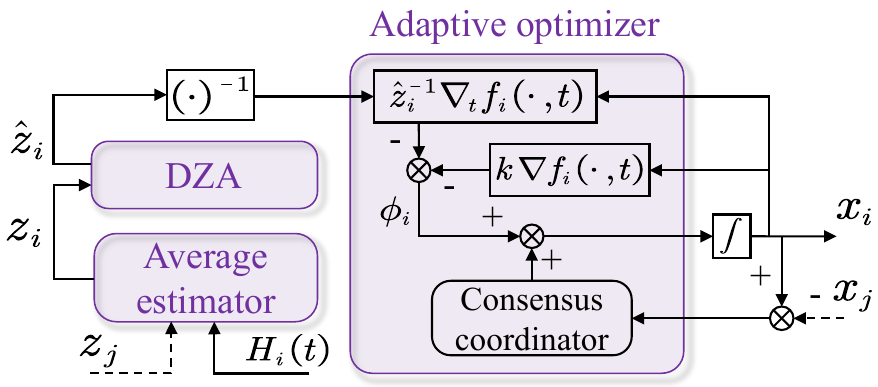}}
    \caption{The structure of the proposed design.}
    \label{fig.control}
\end{figure}

The DZA bridges the average estimator and the adaptive optimizer, enabling implementation without relying on the waiting time (required in \cite{Ren17,Chu22,Wang22,Sun23}). This integration arises a new convergence analysis issue. The dead-zone nonlinearity introduced by the DZA and the inverse operation of its output, $(\cdot)^{-1}$, also cause difficulties in subsequent analysis.

\begin{remark}\label{nonsmooth}
The trajectory of $\hat{z}_i(t)$ generated by Algorithm \ref{alg} is discontinuous but measurable and locally essentially bounded. Thus, the Filippov solutions of the proposed design always exist. Since the Lyapunov functions used in convergence analysis are continuously differentiable, the proof still holds without using the nonsmooth analysis\cite{Cortes08,Sun23}.
\end{remark}

\section{Convergence Analysis}\label{ConvergenceAnalysis}

Define $x=(x_1,\cdots,x_N)^\top$, $\phi=(\phi_1^\top,\cdots,\phi_N^\top)^\top$ and a diagonal matrix $B$, whose order is equal to the number of edges in the graph $\mathcal{G}$. The nodes corresponding to non-zero elements in column $i$ of matrix $D$ are denoted as $i_1$ and $i_2$, respectively. The diagonal element $B_{ii}$ of $B$ is defined as $\beta_{i_1i_2}$. Then, we rewrite \eqref{eq.ui} as
\begin{align}\label{eq.compact}
    \dot{x}&=-\bar{L}_{\otimes}x-(DB)_{\otimes}S(D^{\top}_{\otimes}x)+\phi,
\end{align}
where $\bar{L}=[l^{\alpha}_{ij}(t)]\in\mathbb{R}^{N\times N}$ with $l^{\alpha}_{ii}(t)=\sum_{j\neq i}\alpha_{ij}(t)$ and $l^{\alpha}_{ij}(t)=-\alpha_{ij}(t)$ for $j\neq i$. Define the consensus error $e=M_{\otimes} x$ with $M=I_N-1_N1_N^{\top}/N$, and the error between the average state and the global optimal solution $\tilde{x}=\textbf{1}^{\top}_Nx/N-r^*$ with $\textbf{1}^{\top}_N=1_N^{\top}\otimes I_m$. This implies the following relation between $x$, $e$ and $\tilde{x}$, i.e.,
\begin{align}\label{eq.xetr}
    x=e+1_N\otimes\tilde{x}+1_N\otimes r^*.
\end{align}  
With $\nabla F(x,t)=(\nabla f_1(x_1,t),\cdots,\nabla f_N(x_N,t))^\top$, we obtain 
\begin{align}
\dot{e}=&-\bar{L}_{\otimes}e-(DB)_{\otimes}S(D^{\top}_{\otimes}e)+M_{\otimes}\phi,\label{eq.e}\\
\dot{\tilde{x}}=&-\frac{k}{N}\textbf{1}^{\top}_N\nabla F(x,t)-\frac{1}{N}\sum_{i\in\mathcal{N}}\hat{z}^{-1}_i\nabla_t f_i(x_i,t)-\dot{r}^*,\label{eq.tildex}
\end{align}
where $\dot{r}^*=-{H}^{-1}(t)\sum_{i\in\mathcal{N}}\nabla_t f_i(r^*,t)$. 

The following theorem gives the main result of this note.
\begin{theorem}\label{distributed.oneorder}
Suppose Assumptions \ref{assumption.graph1}-\ref{assumption.boundedness.hessian} hold. For system \eqref{eq.z}-\eqref{eq.alpha} with Algorithm \ref{alg}, let $h_0=\gamma h_1/N$ and $k>h_2N\sqrt{m}/h_1^2+\epsilon_3$ with $0<\gamma<1$ and $\epsilon_3>0$. Then, the optimization problem (\ref{optimizationproblem}) is solved, i.e., $\lim_{t\rightarrow \infty}x_i(t)-r^*(t)=0$, $i\in\mathcal{N}$.
\end{theorem}
\begin{proof} 
See Appendix \ref{appendix.proof3}.
\end{proof}

In this note, the analyses of error subsystems are unified within the common time domain, which differs from that of the waiting time-based design \cite{Ren17,Wang22,Chu22,Sun23}. Specifically, we establish the Lyapunov analyses of the $e$ and $\tilde{x}$ subsystems based on \eqref{eq.xetr}-\eqref{eq.tildex} and Lemma \ref{condition}. With the existing analysis of the $\delta$ subsystem, we consider the proposed algorithm as an interconnected system, as shown in Fig. \ref{fig.subsystem}. By further employing the finite escape-time analysis (Lemma \ref{xbounded} in Appendix \ref{lemmas}), we prove the global asymptotic
convergence to the optimal solution $r^*(t)$.
\begin{figure}[!h]
    \centerline{\includegraphics[width=0.75\columnwidth]
    {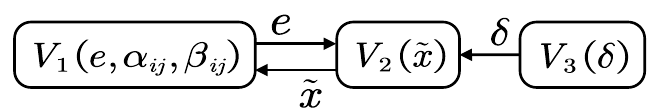}}
    \caption{The interconnections within the proposed algorithm.}
    \label{fig.subsystem}
\end{figure}

\section{Extensions: Non-Quadratic Cost Function}

In this section, we consider the local cost function $f_i(x_i,t)$ of a general form as
\begin{align}\label{eq.ffi}
f_i(x_i,t)=\frac{1}{2}x_i^{\top}H_i(t)x_i+R_i^{\top}(t)x_i+\hat{f}_i(x_i,t),~i\in\mathcal{N},
\end{align} 
where $\hat{f}_i(x_i,t)$ is not limited to the quadratic form. Let $\rho_i(x_i,t)$ denote the Hessian of $f_i(x_i,t)$. We then make the following assumption on $f(r,t)$ and $\hat{f}_i(x_i,t)$.
\begin{assumption}\label{assumption.newboundedness}
For $\forall t\geq0$ and $\forall i\in\mathcal{N}$, there exist $\mu_j>0$, $j=1,\cdots,5$ such that 
\begin{enumerate}[(1)]
    \item \label{a} The global cost function $f(r,t)$ is uniformly $\mu_1$-strongly convex and $\nabla \hat{f}_i(x_i,t)$ is uniformly globally Lipschitz, 
    \item \label{b} $|\nabla\hat{f}_i(x_i,t)|\leq\mu_2|x_i|+\mu_3$,
    \item \label{c}$\nabla_t\hat{f}_i(x_i,t)=g_i(t)x_i+s_i(x_i,t)$ where $|g_i(t)|\leq\mu_4$, $|s_i(x_i,t)|\leq\mu_5\bar{s}(t)<\infty$ and $\int_{0}^{\infty}\bar{s}^2(t) \,dt<\infty$.
\end{enumerate}
\end{assumption}

\begin{remark}
The requirement of uniform strong convexity is still only applied to $f(r,t)$ here. The uniform global Lipschitz property of the gradient function is also adopted in \cite{Chu22,ZOU2020,kia2015}. Assumption \ref{assumption.newboundedness}.(\ref{b}) and \ref{assumption.newboundedness}.(\ref{c}) respectively require $\nabla\hat{f}_i(x_i,t)$ and $\nabla_t\hat{f}_i(x_i,t)$ to be linearly bounded. Similar assumptions are also used in \cite{Ren17,Huang20,Zheng24}. In fact, Assumption \ref{assumption.newboundedness} holds for many cost functions, as shown in Example 1 of Section \ref{section.simulation}.
\end{remark}
The algorithm for solving the optimization problem \eqref{optimizationproblem} with $f_i(x_i,t)$ defined by \eqref{eq.ffi} follows the design process detailed in Section \ref{framework}. However, the estimator \eqref{eq.z} cannot be applied here to estimate $\sum_{i\in \mathcal{N}}\rho_i(x_i,t)/N$ with $\rho_i(x_i,t)$ replacing $H_i(t)$. The reason is that $\dot{\rho}_i(x_i,t)$ may contain $\dot{x}_i$, whose upper bound is unknown before implementing algorithm. Consequently, the fixed gain $\omega$ cannot be determined in advance due to the unknown upper bound of $|\dot{\rho}_i(x_i,t)|_\infty$. To address this issue, we introduce a new average estimator adapted from \cite{Jiang24}: 
\begin{align}\label{eq.newestimator}
    \dot{\xi}_i=&-\sum_{j\in\mathcal{N}_i}\omega_{ij}(t)\sgn(z_i-z_j),~~z_i=\xi_i+\rho_i,~~i\in\mathcal{N},
\end{align}
where $\omega_{ij}>((N-1)/2)(|\dot{\rho}_i|_\infty+|\dot{\rho}_j|_\infty)$ is a {\em state-based} gain. The initial conditions satisfy $\sum_{j\in\mathcal{N}}\xi_i(0)=0$ and $\xi_i(0)$ is symmetric. Similar to the proof in \cite{Jiang24}, $z_i$ can be shown to converge to $\sum_{i\in \mathcal{N}}\rho_i(x_i,t)/N$ in finite time. The next corollary extends Theorem \ref{distributed.oneorder} to address the distributed optimization problem with  more general cost functions \eqref{eq.ffi}.

\begin{corollary}\label{distributed.corollary}
Suppose Assumptions \ref{assumption.graph1}, \ref{assumption.boundedness.hessian}, \ref{assumption.newboundedness} hold. For system \eqref{eq.newestimator} and \eqref{eq.ui}-\eqref{eq.alpha} with Algorithm \ref{alg}, let $h_0=\gamma \mu_1/N$ and $k>(h_2+\mu_4)N\sqrt{m}/\mu_1^2+\epsilon_4$ with $0<\gamma<1$ and $\epsilon_4>0$. Then, the optimization problem (\ref{optimizationproblem}) is solved, i.e., $\lim_{t\rightarrow \infty}x_i(t)-r^*(t)=0$, $i\in\mathcal{N}$.
\end{corollary}
\begin{proof} 
See Appendix \ref{appendix.proof4}. 
\end{proof}
\begin{remark}
    Different from estimator \eqref{eq.z}, the implementation of estimator \eqref{eq.newestimator} does not require a known upper bound on $|\dot{\rho}_i|_\infty$, but at the expense of requiring additional exchange of information $|\dot{\rho}_i|_\infty$ with neighboring agents. 
\end{remark}

\section{Simulations}
\label{section.simulation}

{\bf Example 1.} Consider $20$ agents interacting over a circular communication network. With the decision variable $y=(y_1,y_2,y_3)\in\mathbb{R}^3$, each agent is assigned with a local cost function. For agent $i=1,\cdots,6$, $f_i=\sum_{j=1}^32y_j^2 +0.1i\sin (t) y_j$. For agent $i=7,\cdots,12$, $f_i=(y_1+\cos(t)y_2+0.1iy_3)^2+\cos(t)y_1+0.1iy_2+y_3$. The agent $i=13,\cdots,20$ is assigned with $f_{13}=\sum_{j=1}^3\exp(-t)y_j^2$, $f_{14}=\sum_{j=1}^3\sin(t) y_j^2$, $f_{15}=\sum_{j=1}^3-(y_j-2)^2$, $f_{16}=\sum_{j=1}^3\sin(y_j)$, $f_{17}=\sum_{j=1}^3(\tanh(t)+1)\ln(1+\exp(y_j))$, $f_{18}=\sum_{j=1}^3y_j/(y_j^2+1)$, $f_{19}=\sum_{j=1}^3-(1/(t+1))\exp(-y_j^2)$ and $f_{20}=\sum_{j=1}^3y_j^2+\exp(-t)\ln(1+y_j^2)$, respectively. Note that $f_{19}$ is inspired by the density function of normal distribution, and $f_{17}$ and $f_{20}$ are adapted from \cite{Simonetto20} and \cite{Bastianello24}, respectively. 

In fact, except for $f_{15}$ and $f_{17}$, all local cost functions have their {\em own} global minimum. However, only $f_1$-$f_6$ and $f_{20}$ are uniformly strongly convex. $f_{13}$ and $f_{17}$ are uniformly strictly but not strongly convex, $f_7$-$f_{12}$ are just uniformly convex, while $f_{14}, f_{15}, f_{16}, f_{18}$ and $f_{19}$ are even non-convex. It can be checked that Assumptions \ref{assumption.graph1}, \ref{assumption.boundedness.hessian}, \ref{assumption.newboundedness} are satisfied. By applying the distributed algorithm in Corollary \ref{distributed.corollary} with $h_0=0.5$ and $k=35$, the optimization problem \eqref{optimizationproblem} is solved. With $x_i$ corresponding to the decision variable $y$ of agent $i$, Fig. \ref{fig.errorexample1} shows the trajectory of $|x_i-r^*(t)|$ that converges to zero asymptotically. 
\begin{figure}[!h]
    \centerline{\includegraphics[width=1.0\columnwidth]
    {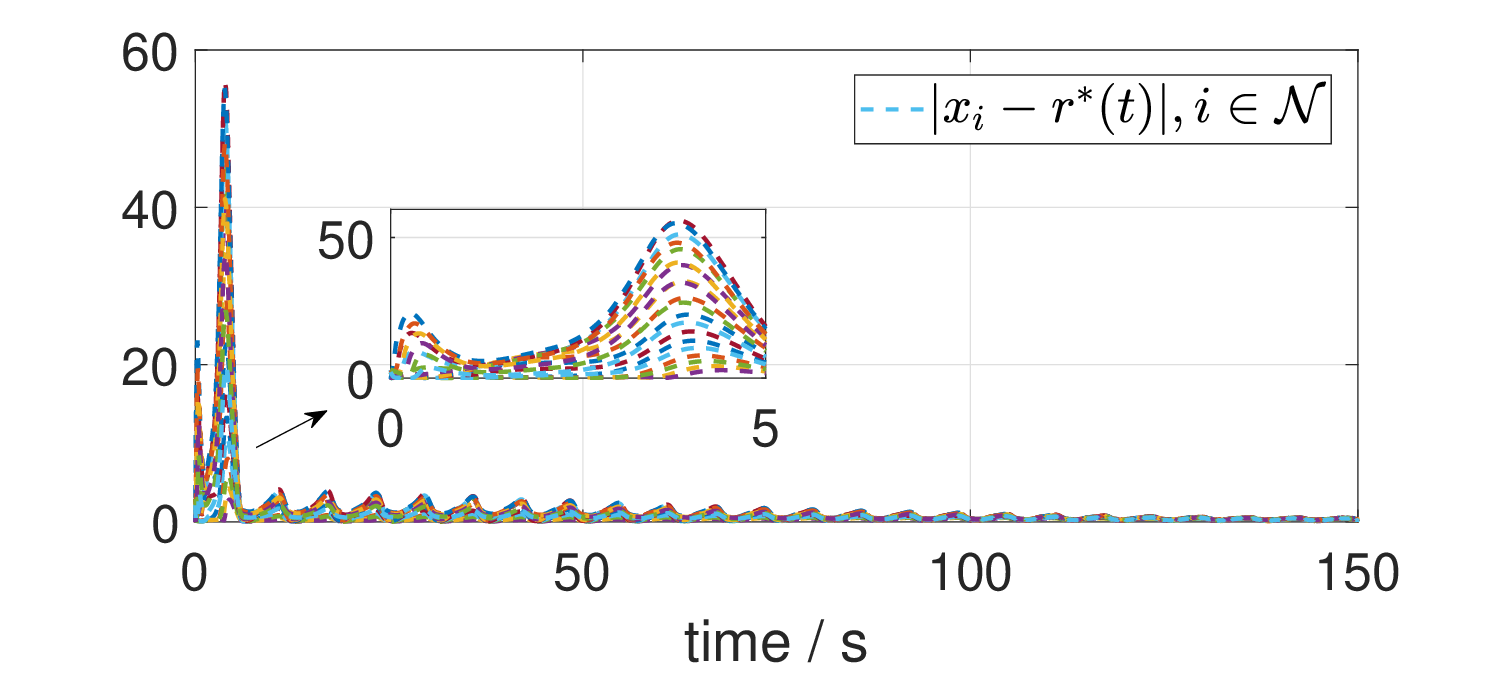}}
    \caption{The trajectory of $|x_i-r^*(t)|$, $i\in\mathcal{N}$.}
    \label{fig.errorexample1}
\end{figure}

{\bf Example 2.} Consider a scenario in Fig. \ref{fig.scenario}, where $N=6$ UAVs are used as relay nodes for a communication system supporting long distance robots collaboration. The signal power of robot $i\in\mathcal{N}$ is modeled by \cite{tse2005wireless} as $g_{i}(x_i,t)=p_i(t)/|x_{i}-d_{i}(t)|^2$, $x_i\neq d_i(t)$, in a 2D scenario, where $p_i(t)>0$ and $d_i(t)\in\mathbb{R}^2$ denotes the signal transmission power and the position of robot $i$, respectively. The communication among these robots is assisted by UAVs with relay capabilities\cite{Merwaday15,EmergencyNetworks}. Meanwhile, the UAVs are required to form a formation to expand communication coverage. \begin{figure}[!t]
\centerline{\includegraphics[width=0.9\columnwidth]
    {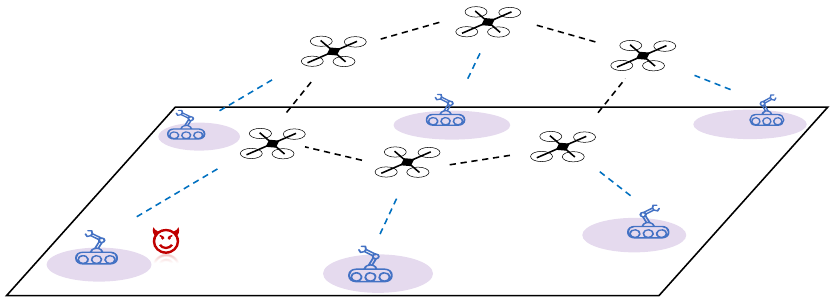}}
    \caption{Multi-UAV relay communication. Purple ellipses denote robot working ranges. The devil icon denotes an attacker.}
    \label{fig.scenario}
\end{figure}
To establish reliable communication links, all UAVs need to search for positions as close as possible to the corresponding robots. This leads to the optimization problem:
\begin{align}\label{eq.example2}
    \min_{x_1,\cdots,x_N}\sum_ {i\in \mathcal{N}}{f_i(x_i,t)},~~~~\text {s.t.} ~~x_i-x_j=\tau_i-\tau_j,
\end{align}
where $f_i(x_i,t)=g_i^{-1}(x_i,t)$ and $\tau_i\in\mathbb{R}^2$ denotes a desired formation. Each UAV $i$ can obtain the real-time parameters $p_i(t)$ and $d_i(t)$ of robot $i$ to construct $f_i(x_i,t)$. The optimal solution of optimization problem \eqref{eq.example2} is $x^*_i=r^*+\tau_i$, $i\in \mathcal{N}$, with $r^*$ being the optimal solution of $\min\sum_ {i\in \mathcal{N}}{f_i(r+\tau_i,t)}$. Then, $x^*_i$ can be regarded as the optimal relay position of UAV $i$. Moreover, the following two practical cases are considered:
\begin{enumerate}[(1)]
    \item {The signal transmission power of a robot falls below $p_0$,}
    \item {The attacker's actions damage a robot at time instant $t_0$.}
\end{enumerate}
In both cases, the robots lose the ability to transmit their parameters $p_i(t)$ and $d_i(t)$ to UAVs, preventing the UAVs from constructing $f_i(x_i,t)$. To simulate this, the corresponding local cost function is set to zero\footnote{To ensure almost everywhere differentiability of parameter functions, we use a low-pass filter to attenuate the coefficient $1/p_i$ to $0$, thereby attenuating the cost function to $0$. Although there are non-differentiable points, the proof in this note still holds without using the nonsmooth analysis, as the Lyapunov functions employed are continuously differentiable.}, which is non-strongly convex.
 
As an example, suppose that the parameters of power model (SI units are used here) are given by $p_1=p_2=10$, $p_3=10\exp(-0.05t)$, $p_4=8$, $p_5=10/(1+0.06t)$, $p_6=6+\sin(0.3t)$, $d_1=(30,30)^{\top}$, $d_2=(30+10\cos(0.3t),0)^{\top}$, $d_3=(-30,-30)^{\top}$, $d_4=(-30,-30+10\sin(0.3t))^{\top}$, $d_5=(-30,0)^{\top}$ and $d_6=(-30,30)^{\top}$. Let $p_0=2$. The desired formation is given by a regular hexagon as $\tau_1=(5,10)^{\top}$, $\tau_2=(10,0)^{\top}$, $\tau_3=(5,-10)^{\top}$, $\tau_4=(-5,10)^{\top}$, $\tau_5=(-10,0)^{\top}$ and $\tau_6=(-5,10)^{\top}$. The initial states are given by $x_1(0)=(-5,10)^{\top}$, $x_2(0)=(-15,10)^{\top}$, $x_3(0)=(-25,10)^{\top}$, $x_4(0)=(-25,0)^{\top}$, $x_5(0)=(-15,0)^{\top}$ and $x_6(0)=(-5,0)^{\top}$. Following the aforementioned two cases, we know $p_3<p_0$ if $t\geq 32$s. Moreover, assume that robot $5$ is damaged by the attacker at $t=t_0=67$s. A circular communication network $\mathcal{G}$ is used for UAVs. One may check that Assumptions \ref{assumption.graph1}-\ref{assumption.boundedness.hessian} are satisfied. Then, each UAV modeled as a single-integrator dynamics is driven by the proposed distributed algorithm in Theorem \ref{distributed.oneorder} with $h_0=0.1$, $\omega=1$ and $k=20$ to search the optimal relay position. Fig. \ref{fig.position} shows the trajectories of $x_i$, $i\in\mathcal{N}$, $r^*(t)$ and the snapshots of them at $t=40$s, $90$s, $140$s. Fig. \ref{fig.error} shows the trajectory of $ |x_i-\tau_i-r^*(t)|$, $i\in\mathcal{N}$. It can be seen that even if $f_3=0$ when $t\geq 32$s and $f_5=0$ when $t\geq 67$s, the asymptotical convergence of $x_i$ to $\tau_i+r^*(t)$ still be achieved.
\begin{figure}[!h]
\centerline{\includegraphics[width=1.0\columnwidth]
    {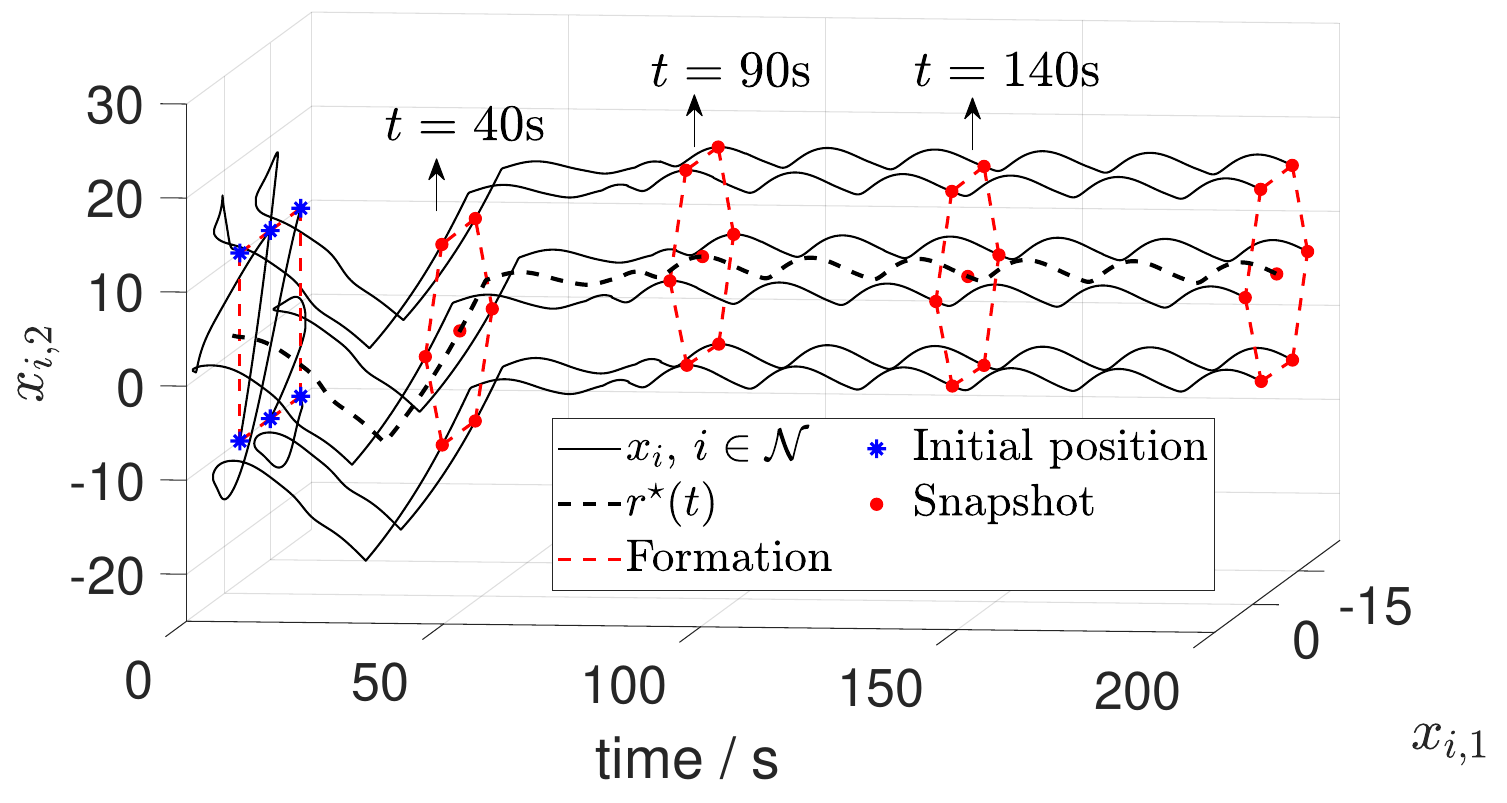}}
    \caption{The trajectories of $x_i$, $i\in\mathcal{N}$, and $r^*(t)$.}
    \label{fig.position}
\end{figure}
\begin{figure}[!h]
\centerline{\includegraphics[width=1.0\columnwidth]
    {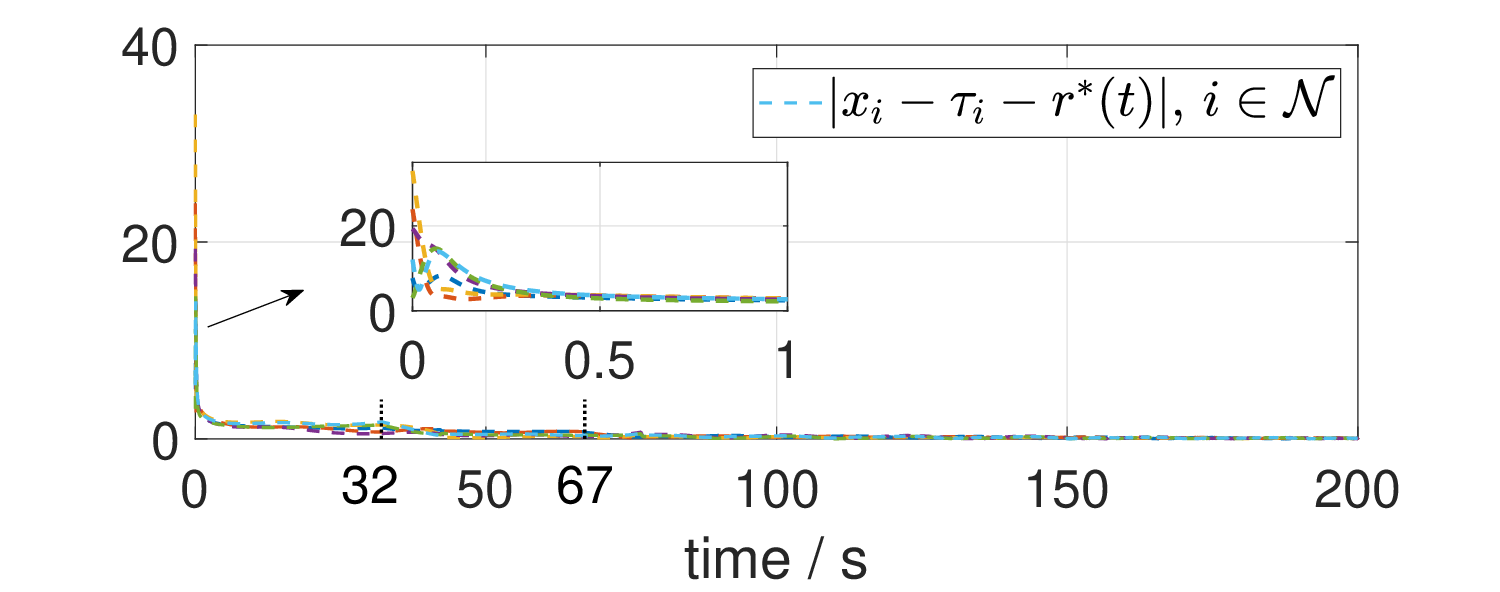}}
    \caption{The trajectory of $ |x_i-\tau_i-r^*(t)|$, $i\in\mathcal{N}$.}
    \label{fig.error}
\end{figure}
 
\section{Conclusion}
\label{Conclusion}
This note first proposes a systematic design incorporating an average estimator, the DZA and an adaptive optimizer to solve DTQO. It relaxes the assumptions that positive definite, identical, diagonal and time-invariant Hessians of local cost functions. It does not require the waiting time during algorithm implementation, nor certain prior knowledge of the communication graph and global parameter of the cost functions. It can be shown that the global asymptotic convergence to the optimal solution. The result is also extended to a class of non-quadratic cost functions. This advancement enables large-scale multi-robot systems to achieve autonomous optimization more flexibly in dynamic environments.

\appendix

\subsection{ Two Auxiliary Lemmas}\label{lemmas}
\begin{lemma}\label{mirsky}
    \cite[Corollary 7.4.9.3]{matrix}: Let $A,B\in\mathbb{R}^{m\times m}$ be Hermitian and let $|\cdot|_u$ be an unitarily invariant norm on $\mathbb{R}^{m\times m}$. Then, $|\diag\lambda^\downarrow(A)-\diag\lambda^\downarrow(B)|_u\leq|A-B|_u$, where $\diag\lambda^\downarrow(A)$($\diag\lambda^\downarrow(B)$) is the diagonal matrix whose diagonal entries are the nonincreasingly ordered eigenvalues of $A$($B$).
\end{lemma}

\begin{lemma}\label{xbounded}
    Suppose Assumptions \ref{assumption.graph1}-\ref{assumption.boundedness.hessian} hold. For system \eqref{eq.z}-\eqref{eq.alpha} with Algorithm \ref{alg}, the state $x_i(t)$ does not have a finite escape time, i.e., for any $\bar{t}\in[0,\infty)$  the state $x_i(t)$ is bounded for all $t\in[0,\bar{t}]$. 
    \end{lemma}
        
    \begin{proof}
    Denote $x_{ip}\in\mathbb{R}$ the $p$-th row element of $x_{i}$, where $i\in\mathcal{N}$ and $p\in\mathcal{M}$. Similar subscripts also appear in certain matrices that follow. We take $x_{ip}$ as an example to check its boundedness in a finite time. The dynamics of $x_{ip}$ is given by
    \begin{align}\label{eq.fenliang}
    \dot{x}_{ip}=&-kH_{ip}x_{i}-kR_{ip}-{\hat{z}_{ip}}^{-1}\dot{H}_{i}x_{i}-{\hat{z}_{ip}}^{-1}\dot{R}_{i}\nonumber\\
    &+c_1(t)\sum_{j\in\mathcal{N}_{i}}S(x_{ip}-x_{jp})+W_1+W_2,
    \end{align}
    where $0\leq c_1(t)=\int_{0}^{t}m\epsilon_1 \eta_t\,dt=(m\epsilon_1/\epsilon_2)(1-\eta_t)<\infty$, $W_1=-\sum_{j\in\mathcal{N}_{i}}(\int_{0}^{t} |x_{i}-x_j|^2 \,dt) (x_{ip}-x_{jp})$, and $W_2=-\sum_{j\in\mathcal{N}_{i}}(\int_{0}^{t}|x_{i}-x_j|_1\,dt)S(x_{ip}-x_{jp})$. We first consider the first half of \eqref{eq.fenliang}, i.e.
    \begin{align}\label{eq.example}
    \dot{x}_{ip}=&-kH_{ip}x_{i}-kR_{ip}-{\hat{z}_{ip}}^{-1}\dot{H}_{i}x_{i}-{\hat{z}_{ip}}^{-1}\dot{R}_{i}\nonumber\\
    &+c_1(t)\sum_{j\in\mathcal{N}_{i}}S(x_{ip}-x_{jp})=:g_i(x_{ip},t).
    \end{align}
    For all $t\in[0,\bar{t}]$, we know $g_i(x_{ip},t)$ is piecewise continuous with respect to $t$. With Assumptions \ref{assumption.graph1}-\ref{assumption.boundedness.hessian} satisfied, the boundedness of $H_i$, $\dot{H}_i$, $\hat{z}_{i}^{-1}$ can be guaranteed, which implies that $g_i(x_{ip},t)$ is global Lipschitz with respect to $x_{ip}$. Then, by \cite[Theorem 3.2]{Nonlinear02khalil}, we know \eqref{eq.example} has a unique solution over $[0,\bar{t}]$. In view of the fact that $\bar{t}$ can be arbitrarily large, we make the following statement,
    \begin{itemize}
        \item [\bf S1.] $x_{ip}$ in system \eqref{eq.example} does not have a finite escape time, $\forall i\in\mathcal{N}$ and $\forall p\in\mathcal{M}$.
    \end{itemize}
    We further consider $W_1$ and $W_2$ in system \eqref{eq.fenliang}. Assume that $i^*,p^*\in\mathop{\arg\max}_{i\in\mathcal{N},p\in\mathcal{M}}|x_{ip}|$. Without loss of generality, let $x_{i^*p^*}>0$. If $x_{i^*}=x_j$, $\forall j\in\mathcal{N}_{i^*}$, we have $W_1=W_2=0$. If there exists at least one $j\in\mathcal{N}_{i^*}$ such that $x_{i^*}\neq x_j$, it follows that $W_1,W_2\leq 0$. Similarly, assume that $x_{i^*p^*}<0$. If $x_{i^*}=x_j$, $\forall j\in\mathcal{N}_{i^*}$, we have $W_1=W_2=0$. If there exists at least one $j\in\mathcal{N}_{i^*}$ such that $x_{i^*}\neq x_j$, then we have $W_1,W_2\geq0$. Therefore, we make another statement,
    \begin{itemize}
        \item [\bf S2.] $W_1+W_2\leq0$ if $x_{i^*p^*}>0$ and $W_1+W_2\geq0$ if $x_{i^*p^*}<0$.
    \end{itemize}   
    
By combining statements {\bf S1} and {\bf S2}, we conclude from \eqref{eq.fenliang} that the trajectory of $x_{i^*p^*}$ does not escape to infinity at a finite time. With the fact that $i^*,p^*\in\mathop{\arg\max}_{i\in\mathcal{N},p\in\mathcal{M}}|x_{ip}|$, thus, $x_i(t)$ does not have a finite escape time, $\forall i\in\mathcal{N}$.
\end{proof} 

\subsection{Proof of Lemma \ref{condition}}\label{appendix.proof2}
Define $\delta_i=z_i(t)-\bar{H}(t)$. We next establish the relation between $|\hat{z}_i-\bar{H}|$ and $|\delta_i|$ in the following two cases:

{\em (1) The non-singular case of matrix $z_i(t)$}. Recall Algorithm \ref{alg}. For all $t\geq 0$, Once $\mathop{\min}_{p\in\mathcal{M}}|\lambda_p(z_i(t))|\geq h_0$, we have $\hat{z}_i(t)=z_i(t)$, implying $|\hat{z}_i(t)-\bar{H}|=|\delta_i|$.

{\em (2) The (near-)singular case of matrix $z_i(t)$}.
For all $t\geq 0$, if $\mathop{\min}_{p\in\mathcal{M}}|\lambda_p(z_i(t))|<h_0$, we have
\begin{align}\label{eq.min}
\mathop{\min}_{p\in\mathcal{M}}|\lambda_p(z_i(t))|<h_0=\gamma h_1/N< {\lambda_1}(\bar{H}(t)),
\end{align}
where the last inequality follows the fact that $\lambda_1(\bar{H}(t))\geq h_1/N$ given by Assumption \ref{assumption.lipschitz.stronglyconvex}. Before establishing the relation between $|\hat{z}_i-\bar{H}|$ and $|\delta_i|$, we first show the symmetry of $z_i$. Rewrite \eqref{eq.z} as $\dot{z}_i=-\omega\sum_{j\in\mathcal{N}_i}$ $\sgn(z_i-z_j)+\dot{H}_i(t)$. Since $H_i(t)$ is symmetric, $\dot{H}_i(t)$ must also be symmetric, $\forall t\geq0$. Symmetric $\xi_i(0)$ indicates that $z_i(0)$ is symmetric as well. It is clear that the function $\sgn(\cdot)$ does not change the symmetry of its matrix independent variable. Thus, $z_i(t)$ is always symmetric. By Lemma \ref{mirsky} along with \eqref{eq.min}, we conclude that $|\delta_i|\geq (1-\gamma)h_1/N>0$. Consequently, we can always find a positive constant $k^s_i$ such that $|\hat{z}_i(t)-\bar{H}|\leq k^s_i|\delta_i|$.

Now, based on the analysis of the above two case, we have
$|\hat{z}_i(t)-\bar{H}|\leq \max\{1,k^s_i\}|\delta_i|={k}_1|\delta_i|$, $\forall t\geq 0$, where ${k}_1=\max_{i\in\mathcal{N}}\{1,k^s_i\}$. This means $|\hat{z}_i-\bar{H}|$ can always be linearly bounded by $|\delta_i|$, whether $z_i(t)$ is singular or not.

\subsection{Proof of Theorem \ref{distributed.oneorder}}\label{appendix.proof3}
Recall the definition of $\delta_i$ in Appendix \ref{appendix.proof2}, and let $\delta=(\delta_1,\cdots,\delta_N)^\top$. The convergence analysis of the $\delta$ subsystem has been established by \cite{fei12} and is thus omitted here. The following proof focuses on the analysis of $e$ and $\tilde{x}$ subsystems, and the comprehensive analysis of the closed-loop system.

{\em{(1)} Analysis of $e$ subsystem.} Consider the function
\begin{align}
V_1=e^{\top}e+\frac{1}{2}\sum_{i\in\mathcal{N}}\sum_{j\in \mathcal{N}_i}(\alpha_{ij}-\bar{\alpha})^2+\frac{1}{2}\sum_{i\in\mathcal{N}}\sum_{j\in \mathcal{N}_i}(\beta_{ij}-\bar{\beta})^2,  \nonumber
\end{align}
where $\bar{\alpha}$ and $\bar{\beta}$ are positive constants to be determined. Its derivative along \eqref{eq.alpha} and \eqref{eq.e} is given by
\begin{align}\label{eq.dotv1}
    \dot{V}_1=&\sum_{i\in\mathcal{N}}\sum_{j\in \mathcal{N}_i}(\alpha_{ij}-\bar{\alpha})\dot{\alpha}_{ij}+\sum_{i\in\mathcal{N}}\sum_{j\in \mathcal{N}_i}(\beta_{ij}-\bar{\beta})\dot{\beta}_{ij}\nonumber\\
    &-2e^{\top}\bar{L}_{\otimes}e+\Lambda_3+\Lambda_4,
\end{align}
where $\Lambda_3=-2e^{\top}(DB)_{\otimes}S(D^{\top}_{\otimes}e)$ and $\Lambda_4=2e^{\top}M_{\otimes}\phi$. Note
\begin{align}\label{eq.ese}
\Lambda_3=&-2\sum_{i\in\mathcal{N}}\sum_{j\in\mathcal{N}_i}\beta_{ij}e_i^{\top}S(e_i-e_j)\nonumber\\
=&-\sum_{i\in\mathcal{N}}\sum_{j\in\mathcal{N}_i}\beta_{ij}(e_i-e_j)^{\top}S(e_i-e_j)\nonumber\\
\leq& -\sum_{i\in\mathcal{N}}\sum_{j\in\mathcal{N}_i}\beta_{ij}(|e_i-e_j|_1-m\epsilon_1 \eta_t),
\end{align}
where the second equality follows from Assumption \ref{assumption.graph1} and and third one follows the fact that $\beta_{ij}\geq0$ and $yS(y)\geq|y|-\epsilon_1 \eta_t$ with $y\in\mathbb{R}$. By \eqref{eq.phi}, we obtain $|\phi|\leq {\phi}_1|x|+{\phi}_2$ with ${\phi}_1=|\diag\{kH_i+\hat{z}_{i}^{-1}\dot{H}_i\}|$, ${\phi}_2=|((kR_1+\hat{z}_{1}^{-1}\dot{R}_1)^\top,\cdots,(kR_N+\hat{z}_{N}^{-1}\dot{R}_N)^\top)^\top|$. The boundedness of $H_i$, $\dot{H}_i$, $\hat{z}_{i}^{-1}$, $R_i$ and $\dot{R}_i$ implies the boundedness of ${\phi}_1$ and ${\phi}_2$. By denoting $\bar{m}=|M_{\otimes}|$, we have
\begin{align}\label{eq.ephi}
\Lambda_4\leq&2\bar{m}{\phi}_1|e|^2+2\bar{m}{\phi}_1\sqrt{N}|e||\tilde{x}|+2\bar{m}({\phi}_1\sqrt{N}|r^*|+{\phi}_2)|e|\nonumber\\
\leq&b_1|e|^2+\frac{\sigma_1h_1}{N}|\tilde{x}|^2+2\bar{m}({\phi}_1\sqrt{N}|r^*|+{\phi}_2)|e|,
\end{align}
where $\sigma_1>0$, $b_1=2\bar{m}{\phi}_1+\bar{m}^2{\phi}_1^2N^2/(\sigma_1h_1)$, the first inequality is obtained by \eqref{eq.xetr}, and the second inequality follows from Young's inequality \cite{KKKbook}. By Assumption \ref{assumption.graph1} again, we have
\begin{align}\label{eq.alphaij}
\sum_{i\in\mathcal{N}}\sum_{j\in \mathcal{N}_i}(\alpha_{ij}-\bar{\alpha})\dot{\alpha}_{ij}=&\sum_{i\in\mathcal{N}}\sum_{j\in \mathcal{N}_i}(\alpha_{ij}-\bar{\alpha})|e_i-e_j|^2\nonumber\\
=&2e^{\top}(\bar{L}-\bar{\alpha}L)_{\otimes}e.
\end{align}
Substituting \eqref{eq.alpha}, \eqref{eq.ese}, \eqref{eq.ephi} and \eqref{eq.alphaij} into \eqref{eq.dotv1} yields
\begin{align}
\dot{V}_1\leq& -2\bar{\alpha}e^{\top}L_{\otimes}e+b_1|e|^2-\bar{\beta}\sum_{i\in\mathcal{N}}\sum_{j\in\mathcal{N}_i}|e_i-e_j|_1+\frac{\sigma_1h_1}{N}|\tilde{x}|^2\nonumber\\
&+2\bar{m}({\phi}_1\sqrt{N}|r^*|+{\phi}_2)|e|+m\epsilon_1 N^2\bar{\beta}\eta_t.
\end{align}
With the fact that $e^{\top}(1_N\otimes I_m)\equiv 0$, we conclude from \cite{Olfati} that $-e^{\top}L_{\otimes}e\leq-\lambda_2(L)|e|^2$. It can be shown that
\begin{align}\label{eq.ee}
&-\sum_{i\in\mathcal{N}}\sum_{j\in\mathcal{N}_i}|e_i-e_j|_1\nonumber\\
=&-2|D^{\top}_{\otimes}e|_1\leq-2\sqrt{e^{\top}(DD^{\top})_{\otimes}e}\leq-2\sqrt{\lambda_2(L)}|e|,
\end{align}
where the last inequality follows from $L=DD^{\top}$. By selecting $\bar{\beta}=\bar{m}({\phi}_1\sqrt{N}|r^*|+{\phi}_2)/\sqrt{\lambda_2(L)}$, we have
\begin{align}\label{eq.dotv1boound}
\dot{V}_1\leq& -2\lambda_2(L)\bar{\alpha}|e|^2+b_1|e|^2+\frac{\sigma_1h_1}{N}|\tilde{x}|^2+m\epsilon_1 N^2\bar{\beta}\eta_t.
\end{align}

{\em{(2)} Analysis of $\tilde{x}$ subsystem.} Inspired by the form of $\hat{z}^{-1}_i\nabla_t f_i(x_i,t)$, adding $\pm {H}^{-1}\sum_{i\in\mathcal{N}}\nabla_t f_i(x_i,t)$ in \eqref{eq.tildex} yields
\begin{align}
\dot{\tilde{x}}=&-\frac{k}{N}\textbf{1}^{\top}_N\nabla F(x,t)+\Lambda_1+\Lambda_2,
\end{align}
where $\Lambda_1={H}^{-1}\sum_{i\in\mathcal{N}}\nabla_t f_i(r^*,t)-\nabla_t f_i(x_i,t)$ and $\Lambda_2=\sum_{i\in\mathcal{N}}(\bar{H}^{-1}-\hat{z}^{-1}_i)\nabla_t f_i(x_i,t)/N$. The latter follows from $H^{-1}=\bar{H}^{-1}/N$. We find that $\Lambda_1$ can be bounded by $e$ and $\tilde{x}$, i.e.,
\begin{align}\label{eq.lambda1}
\Lambda_1=-{H}^{-1}\sum_{i\in\mathcal{N}}\dot{H}_i(e_i+\tilde{x})\leq \frac{h_2\sqrt{mN}}{h_1}|e+1_N\otimes\tilde{x}|.
\end{align}
Note that 
\begin{align}
    \bar{H}^{-1}-\hat{z}^{-1}_i=\bar{H}^{-1}(\hat{z}_i-\bar{H})\hat{z}^{-1}_i.
\end{align}
By Lemma \ref{condition}, along with the boundedness of $\bar{H}^{-1}$ and $\hat{z}^{-1}_i$, there must be a positive constant ${k}_2={k}_1|\bar{H}^{-1}|$ $\max_{i\in\mathcal{N}}\{|\hat{z}^{-1}_i|\}$ such that
\begin{align}
\Lambda_2\leq\frac{1}{N}\sum_{i\in\mathcal{N}}{k}_2|\delta_i||\nabla_t f_i(x_i,t)|\leq\frac{1}{N}{k}_3|\delta||x|+\frac{1}{\sqrt{N}}{k}_4|\delta|,\nonumber
\end{align}
where ${k}_3={k}_2\max_{i\in\mathcal{N}}\{|{\dot{H}_i}|\}$, ${k}_4={k}_2\max_{i\in\mathcal{N}}\{|{\dot{R}_i}|\}$ and the second inequality is obtained by $|\nabla_t f_i(x_i,t)|\leq{|\dot{H}_i}||x_i|+{|\dot{R}_i|}$. Consider a Lyapunov function candidate $V_2=|\tilde{x}|^2/2$. Thus,
\begin{align}\label{eq.dotv2}
\dot{V}_2=&-\frac{k}{N}\tilde{x}^{\top}\textbf{1}^{\top}_N\nabla F(x,t)+\tilde{x}^{\top}\Lambda_2+\tilde{x}^{\top}\Lambda_1\nonumber\\
=&-\frac{k}{N}\tilde{x}^{\top}\textbf{1}^{\top}_N\nabla F(1_N\otimes(\tilde{x}+r^*),t)+\tilde{x}^{\top}\Lambda_2-\frac{k}{N}\tilde{x}^{\top}\textbf{1}^{\top}_N\nonumber\\
&\times\left(\nabla F(x,t)-\nabla F(1_N\otimes(\tilde{x}+r^*),t)\right)+\tilde{x}^{\top}\Lambda_1.
\end{align}
Invoking Assumption \ref{assumption.lipschitz.stronglyconvex}, $f(r(t),t)$ is uniformly $h_1$-strongly convex with respect to $r(t)$. This implies
\begin{align}\label{eq.v21}
&-\frac{k}{N}\tilde{x}^{\top}\textbf{1}^{\top}_N\nabla F(1_N\otimes(\tilde{x}+r^*),t)\nonumber\\
\leq&-\frac{k}{N}\tilde{x}^{\top}\left(\nabla f(\tilde{x}+r^*,t)-\nabla f(r^*,t)\right)\leq-\frac{kh_1}{N}|\tilde{x}|^2.
\end{align}
By Assumption \ref{assumption.boundedness.hessian}, $\nabla f_i(x_i,t)$ is uniformly $\theta_i$-Lipschitz with $\theta_i=\min_{t\geq0}|H_i(t)|$. By denoting $\theta=\max_{i\in\mathcal{N}}\{\theta_i\}$, we have 
\begin{align}\label{eq.v22}
&-\frac{k}{N}\tilde{x}^{\top}\textbf{1}^{\top}_N\left(\nabla F(x,t)-\nabla F\left(1_N\otimes(\tilde{x}+r^*),t\right)\right)\nonumber\\
\leq& \frac{k\theta}{\sqrt{N}}|\tilde{x}||e|\leq\sigma_2\frac{h_1}{N}|\tilde{x}|^2+b_2|e|^2,
\end{align}
where $\sigma_2>0$ and $b_2=k^2\theta^2/(4\sigma_2h_1)$. We also conclude that
\begin{align}\label{eq.v23}
\tilde{x}^{\top}\Lambda_2\leq&\frac{1}{N}{k}_3|\tilde{x}||\delta||x|+\frac{1}{\sqrt{N}}{k}_4|\tilde{x}||\delta|\nonumber\\
\leq&\frac{{k}_3}{\sqrt{N}}|\delta||\tilde{x}|^2+\frac{{k}_3|\delta|}{N}|\tilde{x}||e|+\frac{{k}_3|r^*|+{k}_4}{\sqrt{N}}|\tilde{x}||\delta|\nonumber\\
\leq&b_3|\delta||\tilde{x}|^2+(\sigma_3+\sigma_4)\frac{h_1}{N}|\tilde{x}|^2+(b_4|e|^2+b_5)|\delta|^2,
\end{align}
where $\sigma_3,\sigma_4>0$, $b_3={k}_3/\sqrt{N}$, $b_4={k}_3^2/(4\sigma_3h_1 N)$, $b_5=({k}_3|r^*|+{k}_4)^2/(4\sigma_4h_1)$. With a positive constant $\sigma_5$, we have 
\begin{align}\label{Lambda1}
\tilde{x}^{\top}\Lambda_1\leq\sigma_5\frac{h_1}{N}|\tilde{x}|^2+b_6|\tilde{x}|^2+b_7|e|^2,
\end{align}
where $b_6=h_2N\sqrt{m}/h_1$ and $b_7=mN^2h_2^2/(4\sigma_5h_1^3)$. Substituting \eqref{eq.v21}, \eqref{eq.v22}, \eqref{eq.v23} and \eqref{Lambda1} into \eqref{eq.dotv2} yields
\begin{align}\label{eq.dotv2bound}
\dot{V}_2\leq&\bigg(\sum_{i=2}^{5}\sigma_i-k\bigg)\frac{h_1}{N}|\tilde{x}|^2+b_6|\tilde{x}|^2+b_3|\delta||\tilde{x}|^2\nonumber\\
&+(b_2+b_7)|e|^2+(b_4|e|^2+b_5)|\delta|^2.
\end{align}

{\em{(3)} Closed-loop synthesis.} Given $V_1$ and $V_2$, we further consider the composite Lyapunov function candidate of $e$ and $\tilde{x}$ subsystems, i.e., $V=V_1+V_2$. Combining \eqref{eq.dotv1boound} and \eqref{eq.dotv2bound} yields
\begin{align}\label{eq.v}
\dot{V}=&\dot{V}_1+\dot{V}_2
\leq-l_1|\tilde{x}|^2-l_2|e|^2+W_3+m\epsilon_1 N^2\bar{\beta}\eta_t,
\end{align}
where $l_1=(k-\sum_{i=1}^{5}\sigma_i)h_1/N-b_6$, $l_2=2\lambda_2(L)\bar{\alpha}-b_1-b_2-b_7$ and $W_3=b_3|\tilde{x}|^2|\delta|+b_4|e|^2|\delta|^2+b_5|\delta|^2$. Then, one may first choose $k>h_2N\sqrt{m}/h_1^2+\epsilon_3$ with $\epsilon_3>0$ such that $l_1>0$. By choosing $\bar{\alpha}>(b_1+b_2+b_7)/(2\lambda_2(L))$, we have $l_2>0$. Integrating both sides of \eqref{eq.v} yields
\begin{align}\label{eq.integrable}
&V(t)+l_1\int_{0}^{\infty} |\tilde{x}|^2 \,dt+l_2\int_{0}^{\infty} |e|^2 \,dt-\int_{0}^{\infty} W_3\,dt\nonumber\\
\leq& V(0)+m\epsilon_1N^2\bar{\beta}/c <\infty.
\end{align}
We next give the boundedness analysis of $\int_{0}^{\infty} W_3\,dt$. By Lemma \ref{xbounded}, $x_i$ is bounded in a finite time, and so is $e$. With the fact that $r^*$ is bounded concluded by Assumption \ref{assumption.boundedness.hessian}, we know $\tilde{x}$ is also bounded in a finite time by \eqref{eq.xetr}. Moreover, by \cite{fei12}, there must be a finite time instant $T$ such that for any $0\leq t<T$, $\delta$ is bounded, and for any $t\geq T$, $\delta_i=z_i(t)-\bar{H}(t)=0$, $\forall i\in\mathcal{N}$. Thus, it is clear that for any $0\leq t<T$, $W_3$ is bounded, which together with the continuity of $e$, $\tilde{x}$ and $\delta$ given by \eqref{eq.e}, \eqref{eq.tildex} and \eqref{eq.z}, respectively, results in that $\int_{0}^{T} W_3\,dt$ is bounded. Note that $\int_{0}^{\infty} W_3\,dt=\int_{0}^{T} W_3\,dt+\int_{T}^{\infty} W_3\,dt=\int_{0}^{T} W_3\,dt$, implying the boundedness of $\int_{0}^{\infty} W_3\,dt$. By recalling \eqref{eq.integrable}, we know $V$ is bounded, and $e,\tilde{x}\in\mathcal{L}_2$. It follows from $V_1$ and $V_2$ that $e$, $\tilde{x}$, $\alpha_{ij}$ and $\beta_{ij}$ are bounded, and thus $x$ is bounded. This implies the boundedness of $\nabla f_i(x_i,t)$ and $\nabla_t f_i(x_i,t)$. Then, by \eqref{eq.e} and \eqref{eq.tildex}, $\dot{e}$ and $\dot{\tilde{x}}$ are bounded. By Barbalat's Lemma \cite{Nonlinear02khalil}, we have $\lim_{t\rightarrow \infty}e(t)=0$ and $\lim_{t\rightarrow \infty}\tilde{x}(t)=0$, which further implies $\lim_{t\rightarrow \infty}{x}_i(t)=r^*(t)$.

\subsection{ Reanalysis of Finite Escape Time for Corollary \ref{distributed.corollary}}

\begin{lemma}\label{co.xbounded}
    Suppose Assumptions \ref{assumption.graph1}, \ref{assumption.boundedness.hessian}, \ref{assumption.newboundedness} hold. For system \eqref{eq.newestimator} and \eqref{eq.ui}-\eqref{eq.alpha} with Algorithm \ref{alg}, the state $x_i(t)$ does not have a finite escape time, i.e., for any $\bar{t}\in[0,\infty)$  the state $x_i(t)$ is bounded for all $t\in[0,\bar{t}]$. 
    \end{lemma}
        
    \begin{proof}
    Denote $x_{ip}\in\mathbb{R}$ the $p$-th row element of $x_{i}$, where $i\in\mathcal{N}$ and $p\in\mathcal{M}$. Similar subscripts also appear in certain matrices that follow. We take $x_{ip}$ as an example to check its boundedness in a finite time. According to Assumption \ref{assumption.newboundedness}.(\ref{b}), without loss of generality, let $\nabla\hat{f}_i(x_i,t)=Y(x_i,t)+Z(x_i,t)$ which satisfies $|Y(x_i,t)|\leq\mu_2|x_i|$ and $|Z(x_i,t)|\leq\mu_3$. Then the dynamics of $x_{ip}$ is given by
    \begin{align}\label{co.eq.fenliang}
    \dot{x}_{ip}=&-k(H_{ip}x_{i}+R_{ip}+Y_{ip}+Z_{ip})\nonumber\\
    &-{\hat{z}_{ip}}^{-1}(\dot{H}_{i}x_{i}+\dot{R}_{i}+g_ix_i+s_i)\nonumber\\
    &+c_1(t)\sum_{j\in\mathcal{N}_{i}}S(x_{ip}-x_{jp})+W_1+W_2,
    \end{align}
    where $0\leq c_1(t)=\int_{0}^{t}m\epsilon_1 \eta_t\,dt=(m\epsilon_1/\epsilon_2)(1-\eta_t)<\infty$, $W_1=-\sum_{j\in\mathcal{N}_{i}}(\int_{0}^{t} |x_{i}-x_j|^2 \,dt) (x_{ip}-x_{jp})$, and $W_2=-\sum_{j\in\mathcal{N}_{i}}(\int_{0}^{t}|x_{i}-x_j|_1\,dt)S(x_{ip}-x_{jp})$. We first consider the first half of \eqref{eq.fenliang}, i.e.
    \begin{align}\label{co.eq.example}
    \dot{x}_{ip}=&-k(H_{ip}x_{i}+R_{ip}+Y_{ip}+Z_{ip})\nonumber\\
    &-{\hat{z}_{ip}}^{-1}(\dot{H}_{i}x_{i}+\dot{R}_{i}+g_ix_i+s_i)\nonumber\\
    &+c_1(t)\sum_{j\in\mathcal{N}_{i}}S(x_{ip}-x_{jp})=:g_i(x_{ip},t).
    \end{align}
    For all $t\in[0,\bar{t}]$, we know $g_i(x_{ip},t)$ is piecewise continuous with respect to $t$. With Assumptions \ref{assumption.graph1}, \ref{assumption.boundedness.hessian}, \ref{assumption.newboundedness} satisfied, we know $g_i(x_{ip},t)$ is global Lipschitz with respect to $x_{ip}$. Then, by \cite[Theorem 3.2]{Nonlinear02khalil}, we know \eqref{co.eq.example} has a unique solution over $[0,\bar{t}]$. In view of the fact that $\bar{t}$ can be arbitrarily large, we make the following statement,
    \begin{itemize}
        \item [\bf S1.] $x_{ip}$ in system \eqref{co.eq.example} does not have a finite escape time, $\forall i\in\mathcal{N}$ and $\forall p\in\mathcal{M}$.
    \end{itemize}
    We further consider $W_1$ and $W_2$ in system \eqref{eq.fenliang}. Assume that $i^*,p^*\in\mathop{\arg\max}_{i\in\mathcal{N},p\in\mathcal{M}}|x_{ip}|$. Without loss of generality, let $x_{i^*p^*}>0$. If $x_{i^*}=x_j$, $\forall j\in\mathcal{N}_{i^*}$, we have $W_1=W_2=0$. If there exists at least one $j\in\mathcal{N}_{i^*}$ such that $x_{i^*}\neq x_j$, it follows that $W_1,W_2\leq 0$. Similarly, assume that $x_{i^*p^*}<0$. If $x_{i^*}=x_j$, $\forall j\in\mathcal{N}_{i^*}$, we have $W_1=W_2=0$. If there exists at least one $j\in\mathcal{N}_{i^*}$ such that $x_{i^*}\neq x_j$, then we have $W_1,W_2\geq0$. Therefore, we make another statement,
    \begin{itemize}
        \item [\bf S2.] $W_1+W_2\leq0$ if $x_{i^*p^*}>0$ and $W_1+W_2\geq0$ if $x_{i^*p^*}<0$.
    \end{itemize}   
    
By combining statements {\bf S1} and {\bf S2}, we conclude from \eqref{eq.fenliang} that the trajectory of $x_{i^*p^*}$ does not escape to infinity at a finite time. With the fact that $i^*,p^*\in\mathop{\arg\max}_{i\in\mathcal{N},p\in\mathcal{M}}|x_{ip}|$, thus, $x_i(t)$ does not have a finite escape time, $\forall i\in\mathcal{N}$.
\end{proof} 

\subsection{Proof of Corollary \ref{distributed.corollary}}\label{appendix.proof4}
Define a diagonal matrix $\Omega$, whose order is equal to the number of edges in the graph $\mathcal{G}$. The nodes corresponding to non-zero elements in column $i$ of matrix $D$ are denoted as $i_1$ and $i_2$, respectively. The diagonal element $\Omega_{ii}$ of $\Omega$ is defined as $\omega_{i_1i_2}$. Then, by defining $z=\col(z_1,\cdots,z_N)$ and  $\rho=\col(\rho_1,\cdots,\rho_2)$, system \eqref{eq.newestimator} can be rewritten as $\dot{z}=-(D\Omega)_{\otimes}\sgn(D^{\top}_{\otimes}z)+\dot{\rho}$. Define consensus error $\bar\delta=M_{\otimes}z$. Then, the dynamics of $\delta$ subsystem can be written as
\begin{align}\label{eq.delta}
\dot{\bar\delta}=&-(D\Omega)_{\otimes}\sgn(D^{\top}_{\otimes}\bar\delta)+M_{\otimes}\dot{\rho}.
\end{align}
{\bf{(1)} Analysis of $\bar\delta$ subsystem.} 
Then, consider the Lyapunov function candidate 
$W=\tr({\bar\delta}^{\top}\bar\delta)$. Its derivative is given by
\begin{align}\label{eq.dotW}
\dot{W}=&-2\tr({\bar\delta}^{\top}(D\Omega)_{\otimes}\sgn(D^{\top}_{\otimes}{\bar\delta}))+2\tr({\bar\delta}^{\top}M_{\otimes}\dot{\rho}).
\end{align}
We first analyze its last term. It can be concluded that
\begin{align}\label{eq.disange1}
2\tr({\bar\delta}^{\top}M_{\otimes}\dot{\rho})=&\frac{1}{N}\sum_{i=1}^{N}\sum_{j=1}^{N}({\bar\delta}_i-{\bar\delta}_j)^{\top}(\dot{\rho}_i-\dot{\rho}_j)\nonumber\\
\leq& \frac{1}{N}\sum_{i=1}^{N}\sum_{j=1}^{N}|{\bar\delta}_i-{\bar\delta}_j|_1|\dot{\rho}_i-\dot{\rho}_j|_{\infty},
\end{align}
which is obtained by Hölder's inequality \cite{Nonlinear02khalil}. By triangular inequality $|\dot{\rho}_i-\dot{\rho}_j|_{\infty}\leq|\dot{\rho}_i|_{\infty}+|\dot{\rho}_j|_{\infty}$, we have
\begin{align}\label{eq.disange3}
2\tr({\bar\delta}^{\top}M_{\otimes}\dot{\rho})\leq& \frac{1}{N}\sum_{i=1}^{N}\sum_{j=1}^{N}(|\dot{\rho}_i|_{\infty}+|\dot{\rho}_j|_{\infty})|{\bar\delta}_i-{\bar\delta}_j|_1\nonumber\\
\leq& \max_i\left\{\sum_{j=1,j\neq i}^{N}(|\dot{\rho}_i|_{\infty}+|\dot{\rho}_j|_{\infty})|{\bar\delta}_i-{\bar\delta}_j|_1\right\}\nonumber\\
\leq& \frac{N-1}{2}\sum_{i\in\mathcal{N}}\sum_{j\in\mathcal{N}_i}(|\dot{\rho}_i|_{\infty}+|\dot{\rho}_j|_{\infty})|{\bar\delta}_i-{\bar\delta}_j|_1,
\end{align}
where the last inequality is obtained by using Assumption \ref{assumption.graph1}. Then, with the similar analysis in \eqref{eq.ese}, we have 
\begin{align}\label{eq.dierge}
-2\tr({\bar\delta}^{\top}(D\Omega)_{\otimes}\sgn(D^{\top}_{\otimes}{\bar\delta}))=-\sum_{i\in\mathcal{N}}\sum_{j\in\mathcal{N}_i}\omega_{ij}||{\bar\delta}_i-{\bar\delta}_j||_1.
\end{align}
Substituting \eqref{eq.disange3} and \eqref{eq.dierge} into \eqref{eq.dotW} yields
\begin{align}
\dot{W}\leq&\sum_{i\in\mathcal{N}}\sum_{j\in\mathcal{N}_i}\left(\frac{N-1}{2}(\bar{\chi}_i+\bar{\chi}_j)-\omega_{ij}\right)||{\bar\delta}_i-{\bar\delta}_j||_1.
\end{align}
With the similar analysis in \eqref{eq.ee}, if we select $\omega_{ij}>((N-1)/2)(|\dot{\rho}_i|_\infty+|\dot{\rho}_j|_\infty)$, we have $\dot{W}\leq-2\sqrt{\lambda_2}W^{1/2}$. According to \cite{Wang10}, there exists a finite time $T_1$ such that for all $ t\geq T_1$, ${\bar\delta}(t)=0$, i.e., $z_i(t)=z_j(t)$, $\forall i,j\in\mathcal{N}$. Moreover, by using Assumption \ref{assumption.graph1}, we have
\begin{align}
    \textbf{1}^{\top}_N\dot{\xi}=-(1^{\top}_ND\Omega)_{\otimes}\sgn(D^{\top}_{\otimes}z)=0.
\end{align}
By recalling the initial condition $\sum_{i\in\mathcal{N}}\xi_i(0)=0$, we have $\sum_{i\in\mathcal{N}}\xi_i(t)=0$. Combining \eqref{eq.newestimator} yields
\begin{align}\label{eq.hengdeng}
\sum_{i\in\mathcal{N}}z_i(t)=\sum_{i\in\mathcal{N}}\rho_i,~~~\forall t\geq 0.
\end{align}
With the fact that $z_i(t)=z_j(t)$, $\forall i,j\in\mathcal{N}$, $\forall t\geq T_1$, we thus have $
z_i=\bar{\rho}:=\sum_{i\in \mathcal{N}}\rho_i(x_i,t)/N$, $\forall t\geq T_1$.

Define $\tilde{z}=\textbf{1}^{\top}_Nz/N-\bar{\rho}$. Then we have
\begin{align}
    z={\bar\delta}+1_N\otimes\tilde{z}+1_N\otimes \bar{\rho}.
\end{align} 
With the fact that $\textbf{1}^{\top}_N{\bar\delta}\equiv 0$, we have 
\begin{align}
    \sum_{i\in\mathcal{N}}z_i(t)=N\tilde{z}+\sum_{i\in\mathcal{N}}\rho_i.
    \end{align}
By recalling \eqref{eq.hengdeng}, we know that $\tilde{z}\equiv0$. This means that ${\bar\delta}_i=z_i-\bar{\rho}=:\delta_i$. Therefore, the finite-time convergence of ${\bar\delta}_i$ is equivalent to the finite-time convergence of $\delta_i$, which will be used for the following analysis.

{\bf{(2)} Analysis of $e$ subsystem.} Consider the function
\begin{align}
V_1=e^{\top}e+\frac{1}{2}\sum_{i\in\mathcal{N}}\sum_{j\in \mathcal{N}_i}(\alpha_{ij}-\bar{\alpha})^2+\frac{1}{2}\sum_{i\in\mathcal{N}}\sum_{j\in \mathcal{N}_i}(\beta_{ij}-\bar{\beta})^2,  \nonumber
\end{align}
where $\bar{\alpha}$ and $\bar{\beta}$ are positive constants to be determined. Its derivative along \eqref{eq.alpha} and \eqref{eq.e} is given by
\begin{align}\label{co.eq.dotv1}
    \dot{V}_1=&\sum_{i\in\mathcal{N}}\sum_{j\in \mathcal{N}_i}(\alpha_{ij}-\bar{\alpha})\dot{\alpha}_{ij}+\sum_{i\in\mathcal{N}}\sum_{j\in \mathcal{N}_i}(\beta_{ij}-\bar{\beta})\dot{\beta}_{ij}\nonumber\\
    &-2e^{\top}\bar{L}_{\otimes}e+\Lambda_3+\Lambda_4,
\end{align}
where $\Lambda_3=-2e^{\top}(DB)_{\otimes}S(D^{\top}_{\otimes}e)$ and $\Lambda_4=2e^{\top}M_{\otimes}\phi$. Note
\begin{align}\label{co.eq.ese}
\Lambda_3=&-2\sum_{i\in\mathcal{N}}\sum_{j\in\mathcal{N}_i}\beta_{ij}e_i^{\top}S(e_i-e_j)\nonumber\\
=&-\sum_{i\in\mathcal{N}}\sum_{j\in\mathcal{N}_i}\beta_{ij}(e_i-e_j)^{\top}S(e_i-e_j)\nonumber\\
\leq& -\sum_{i\in\mathcal{N}}\sum_{j\in\mathcal{N}_i}\beta_{ij}(|e_i-e_j|_1-m\epsilon_1 \eta_t),
\end{align}
where the second equality follows from Assumption \ref{assumption.graph1} and and third one follows the fact that $\beta_{ij}\geq0$ and $yS(y)\geq|y|-\epsilon_1 \eta_t$ with $y\in\mathbb{R}$. By \eqref{eq.phi}, we obtain $|\phi|\leq {\phi}_1|x|+{\phi}_2$ with 
\begin{align}
    \phi_1=&|\diag\{k(H_i+\mu_2I_m)+\hat{z}_{i}^{-1}(\dot{H}_i+\mu_4I_m)\}|,\\
    \phi_2=&\left|\left[ \begin{array}{c}
        k(R_1+\mu_31_m)+\hat{z}_{1}^{-1}(\dot{R}_1+|s(x_1,t)|1_m)\\
        \vdots\\
        k(R_N+\mu_31_m)+\hat{z}_{N}^{-1}(\dot{R}_N+|s(x_N,t)|1_m)\\
        \end{array} \right]\right|.
\end{align}
The boundedness of $H_i$, $\dot{H}_i$, $\hat{z}_{i}^{-1}$, $R_i$, $\dot{R}_i$ and $|s(x_i,t)|$ implies the boundedness of ${\phi}_1$ and ${\phi}_2$. With $\bar{m}:=|M_{\otimes}|$, we have
\begin{align}\label{co.eq.ephi}
\Lambda_4\leq&2\bar{m}{\phi}_1|e|^2+2\bar{m}{\phi}_1\sqrt{N}|e||\tilde{x}|+2\bar{m}({\phi}_1\sqrt{N}|r^*|+{\phi}_2)|e|\nonumber\\
\leq&b_1|e|^2+\frac{\sigma_1\mu_1}{N}|\tilde{x}|^2+2\bar{m}({\phi}_1\sqrt{N}|r^*|+{\phi}_2)|e|,
\end{align}
where $\sigma_1>0$, $b_1=2\bar{m}{\phi}_1+\bar{m}^2{\phi}_1^2N^2/(\sigma_1\mu_1)$, the first inequality is obtained by \eqref{eq.xetr}, and the second inequality follows from Young's inequality \cite{KKKbook}. By Assumption \ref{assumption.graph1} again, we have
\begin{align}\label{co.eq.alphaij}
\sum_{i\in\mathcal{N}}\sum_{j\in \mathcal{N}_i}(\alpha_{ij}-\bar{\alpha})\dot{\alpha}_{ij}=&\sum_{i\in\mathcal{N}}\sum_{j\in \mathcal{N}_i}(\alpha_{ij}-\bar{\alpha})|e_i-e_j|^2\nonumber\\
=&2e^{\top}(\bar{L}-\bar{\alpha}L)_{\otimes}e.
\end{align}
Substituting \eqref{eq.alpha}, \eqref{co.eq.ese}, \eqref{co.eq.ephi} and \eqref{co.eq.alphaij} into \eqref{co.eq.dotv1} yields
\begin{align}
\dot{V}_1\leq& -2\bar{\alpha}e^{\top}L_{\otimes}e+b_1|e|^2-\bar{\beta}\sum_{i\in\mathcal{N}}\sum_{j\in\mathcal{N}_i}|e_i-e_j|_1+\frac{\sigma_1\mu_1}{N}|\tilde{x}|^2\nonumber\\
&+2\bar{m}({\phi}_1\sqrt{N}|r^*|+{\phi}_2)|e|+m\epsilon_1 N^2\bar{\beta}\eta_t.
\end{align}
With the fact that $e^{\top}(1_N\otimes I_m)\equiv 0$, we conclude from \cite{Olfati} that $-e^{\top}L_{\otimes}e\leq-\lambda_2(L)|e|^2$. It can be shown that
\begin{align}
&-\sum_{i\in\mathcal{N}}\sum_{j\in\mathcal{N}_i}|e_i-e_j|_1\nonumber\\
=&-2|D^{\top}_{\otimes}e|_1\leq-2\sqrt{e^{\top}(DD^{\top})_{\otimes}e}\leq-2\sqrt{\lambda_2(L)}|e|,
\end{align}
where the last inequality follows from $L=DD^{\top}$. By selecting $\bar{\beta}=\bar{m}({\phi}_1\sqrt{N}|r^*|+{\phi}_2)/\sqrt{\lambda_2(L)}$, we have
\begin{align}\label{co.eq.dotv1boound}
\dot{V}_1\leq& -2\lambda_2(L)\bar{\alpha}|e|^2+b_1|e|^2+\frac{\sigma_1\mu_1}{N}|\tilde{x}|^2+m\epsilon_1 N^2\bar{\beta}\eta_t.
\end{align}

{\bf{(3)} Analysis of $\tilde{x}$ subsystem.} Denote $\rho=\sum_{i\in \mathcal{N}}\rho_i(x_i,t)$. Here we have $\dot{r}^*=-{\rho}^{-1}(t)\sum_{i\in\mathcal{N}}\nabla_t f_i(r^*,t)$. Then, adding $\pm {\rho}^{-1}\sum_{i\in\mathcal{N}}\nabla_t f_i(x_i,t)$ in \eqref{eq.tildex} yields
\begin{align}
\dot{\tilde{x}}=&-\frac{k}{N}\textbf{1}^{\top}_N\nabla F(x,t)+\Lambda_1+\Lambda_2,
\end{align}
where $\Lambda_1={\rho}^{-1}\sum_{i\in\mathcal{N}}\nabla_t f_i(r^*,t)-\nabla_t f_i(x_i,t)$ and $\Lambda_2=\sum_{i\in\mathcal{N}}(\bar{\rho}^{-1}-\hat{z}^{-1}_i)\nabla_t f_i(x_i,t)/N$. The latter follows from $\rho^{-1}=\bar{\rho}^{-1}/N$. We conclude that
\begin{align}\label{co.eq.lambda1}
\Lambda_1=&-{\rho}^{-1}\sum_{i\in\mathcal{N}}(\dot{H}_i+g_i)(e_i+\tilde{x})+s_i(x_i,t)-s_i(r^*,t)\nonumber\\
\leq& \frac{(h_2+\mu_4)\sqrt{mN}}{\mu_1}|e+1_N\otimes\tilde{x}|+\frac{2N\mu_5}{\mu_1}\bar{s}(t).
\end{align}
Note that $\bar{\rho}^{-1}-\hat{z}^{-1}_i=\bar{\rho}^{-1}(\hat{z}_i-\bar{\rho})\hat{z}^{-1}_i$. By using the similar analysis in Lemma \ref{condition}, we have $|\hat{z}_i(t)-\bar{\rho}|\leq k_1|\delta_i|$, $\forall t\geq 0$. Along with the boundedness of $\bar{\rho}^{-1}$ and $\hat{z}^{-1}_i$, there must be a positive constant ${k}_2={k}_1|\bar{\rho}^{-1}|$ $\max_{i\in\mathcal{N}}\{|\hat{z}^{-1}_i|\}$ such that
\begin{align}
\Lambda_2\leq\frac{1}{N}\sum_{i\in\mathcal{N}}{k}_2|\delta_i||\nabla_t f_i(x_i,t)|\leq\frac{1}{N}{k}_3|\delta||x|+\frac{1}{\sqrt{N}}{k}_4|\delta|,\nonumber
\end{align}
where ${k}_3={k}_2\max_{i\in\mathcal{N}}\{|\dot{H}_i|+|g_i|\}$, ${k}_4={k}_2\max_{i\in\mathcal{N}}\{|\dot{R}_i|+|s_i|\}$ and the second inequality is obtained by 
\begin{align}
    |\nabla_t f_i(x_i,t)|\leq(|\dot{H}_i|+|g_i|)|x_i|+|\dot{R}_i|+|s_i|.
\end{align}
Consider a Lyapunov function candidate $V_2=|\tilde{x}|^2/2$. Thus,
\begin{align}\label{co.eq.dotv2}
\dot{V}_2=&-\frac{k}{N}\tilde{x}^{\top}\textbf{1}^{\top}_N\nabla F(x,t)+\tilde{x}^{\top}\Lambda_2+\tilde{x}^{\top}\Lambda_1\nonumber\\
=&-\frac{k}{N}\tilde{x}^{\top}\textbf{1}^{\top}_N\nabla F(1_N\otimes(\tilde{x}+r^*),t)+\tilde{x}^{\top}\Lambda_2-\frac{k}{N}\tilde{x}^{\top}\textbf{1}^{\top}_N\nonumber\\
&\times\left(\nabla F(x,t)-\nabla F(1_N\otimes(\tilde{x}+r^*),t)\right)+\tilde{x}^{\top}\Lambda_1.
\end{align}
Invoking Assumption \ref{assumption.newboundedness}.(\ref{a}), $f(r(t),t)$ is uniformly $\mu_1$-strongly convex with respect to $r(t)$. This implies
\begin{align}\label{co.eq.v21}
&-\frac{k}{N}\tilde{x}^{\top}\textbf{1}^{\top}_N\nabla F(1_N\otimes(\tilde{x}+r^*),t)\nonumber\\
\leq&-\frac{k}{N}\tilde{x}^{\top}\left(\nabla f(\tilde{x}+r^*,t)-\nabla f(r^*,t)\right)\leq-\frac{k\mu_1}{N}|\tilde{x}|^2.
\end{align}
Moreover, $\nabla f_i(x_i,t)$ is uniformly $\theta_i$-Lipschitz with $\theta_i=\min_{t\geq0}|\rho_i(x_i,t)|$. By denoting $\theta=\max_{i\in\mathcal{N}}\{\theta_i\}$, we have 
\begin{align}\label{co.eq.v22}
&-\frac{k}{N}\tilde{x}^{\top}\textbf{1}^{\top}_N\left(\nabla F(x,t)-\nabla F\left(1_N\otimes(\tilde{x}+r^*),t\right)\right)\nonumber\\
\leq& \frac{k\theta}{\sqrt{N}}|\tilde{x}||e|\leq\sigma_2\frac{\mu_1}{N}|\tilde{x}|^2+b_2|e|^2,
\end{align}
where $\sigma_2>0$ and $b_2=k^2\theta^2/(4\sigma_2\mu_1)$. We also conclude that
\begin{align}\label{co.eq.v23}
\tilde{x}^{\top}\Lambda_2\leq&\frac{1}{N}{k}_3|\tilde{x}||\delta||x|+\frac{1}{\sqrt{N}}{k}_4|\tilde{x}||\delta|\nonumber\\
\leq&\frac{{k}_3}{\sqrt{N}}|\delta||\tilde{x}|^2+\frac{{k}_3|\delta|}{N}|\tilde{x}||e|+\frac{{k}_3|r^*|+{k}_4}{\sqrt{N}}|\tilde{x}||\delta|\nonumber\\
\leq&b_3|\delta||\tilde{x}|^2+(\sigma_3+\sigma_4)\frac{\mu_1}{N}|\tilde{x}|^2+(b_4|e|^2+b_5)|\delta|^2,
\end{align}
where $\sigma_3,\sigma_4>0$, $b_3={k}_3/\sqrt{N}$, $b_4={k}_3^2/(4\sigma_3\mu_1 N)$, $b_5=({k}_3|r^*|+{k}_4)^2/(4\sigma_4\mu_1)$. With $\sigma_5,\sigma_6>0$, we have 
\begin{align}\label{co.Lambda1}
\tilde{x}^{\top}\Lambda_1\leq\sigma_5\frac{\mu_1}{N}|\tilde{x}|^2+b_6|\tilde{x}|^2+b_7|e|^2+\sigma_6\frac{\mu_1}{N}|\tilde{x}|^2+b_8\bar{s}^2(t),
\end{align}
where $b_6=(h_2+\mu_4)N\sqrt{m}/\mu_1$, $b_7=mN^2(h_2+\mu_4)^2/(4\sigma_5\mu_1^3)$ and $b_8=2N^2\mu_5/(4\sigma_6\mu_1^2)$. Substituting \eqref{co.eq.v21}, \eqref{co.eq.v22}, \eqref{co.eq.v23} and \eqref{co.Lambda1} into \eqref{co.eq.dotv2} yields
\begin{align}\label{co.eq.dotv2bound}
\dot{V}_2\leq&\bigg(\sum_{i=2}^{6}\sigma_i-k\bigg)\frac{\mu_1}{N}|\tilde{x}|^2+b_6|\tilde{x}|^2+b_3|\delta||\tilde{x}|^2\nonumber\\
&+(b_2+b_7)|e|^2+(b_4|e|^2+b_5)|\delta|^2+b_8\bar{s}^2(t).
\end{align}

{\bf{(4)} Closed-loop synthesis.} Given $V_1$ and $V_2$, we further consider the composite Lyapunov function candidate of $e$ and $\tilde{x}$ subsystems, i.e., $V=V_1+V_2$. Combining \eqref{co.eq.dotv1boound} and \eqref{co.eq.dotv2bound} yields
\begin{align}\label{co.eq.v}
\dot{V}=&\dot{V}_1+\dot{V}_2
\leq-l_1|\tilde{x}|^2-l_2|e|^2+b_8\bar{s}^2+W_3+m\epsilon_1 N^2\bar{\beta}\eta_t,
\end{align}
where $l_1=(k-\sum_{i=1}^{6}\sigma_i)\mu_1/N-b_6$, $l_2=2\lambda_2(L)\bar{\alpha}-b_1-b_2-b_7$ and $W_3=b_3|\tilde{x}|^2|\delta|+b_4|e|^2|\delta|^2+b_5|\delta|^2$. Then, one may first choose $k>(h_2+\mu_4)N\sqrt{m}/\mu_1^2+\epsilon_4$, where $\epsilon_4>0$ such that $l_1>0$ and choose $\bar{\alpha}>(b_1+b_2+b_7)/(2\lambda_2(L))$ such that $l_2>0$. Integrating both sides of \eqref{co.eq.v} yields
\begin{align}\label{co.eq.integrable}
&V(t)+l_1\int_{0}^{\infty} |\tilde{x}|^2 \,dt+l_2\int_{0}^{\infty} |e|^2 \,dt\nonumber\\
&-b_8\int_{0}^{\infty} \bar{s}^2(t)\,dt -\int_{0}^{\infty} W_3\,dt\leq V(0)+m\epsilon_1N^2\bar{\beta}/c <\infty.
\end{align}
According to Assumption \ref{assumption.newboundedness}.(\ref{c}), we know $\int_{0}^{\infty} \bar{s}^2(t)\,dt$ is bounded. We next give the boundedness analysis of $\int_{0}^{\infty} W_3\,dt$. By Lemma \ref{co.xbounded}, $x_i$ is bounded in a finite time, and so is $e$. With the fact that $r^*$ is bounded concluded by Assumption \ref{assumption.boundedness.hessian}, we know $\tilde{x}$ is also bounded in a finite time by \eqref{eq.xetr}. Moreover, according to {\bf{(1)} Analysis of $\delta$ subsystem}, there exies a finite time instant $T_1$ such that for any $0\leq t<T_1$, $\delta$ is bounded, and for any $t\geq T_1$, $\delta_i=z_i(t)-\bar{\rho}(t)=0$, $\forall i\in\mathcal{N}$. Thus, it is clear that for any $0\leq t<T_1$, $W_3$ is bounded, which together with the continuity of $e$, $\tilde{x}$ and $\delta$ given by \eqref{eq.e}, \eqref{eq.tildex} and \eqref{eq.newestimator}, respectively, results in that $\int_{0}^{T_1} W_3\,dt$ is bounded. Note that $\int_{0}^{\infty} W_3\,dt=\int_{0}^{T_1} W_3\,dt+\int_{T_1}^{\infty} W_3\,dt=\int_{0}^{T_1} W_3\,dt$, implying the boundedness of $\int_{0}^{\infty} W_3\,dt$. By recalling \eqref{co.eq.integrable}, we know $V$ is bounded, and $e,\tilde{x}\in\mathcal{L}_2$. It follows from $V_1$ and $V_2$ that $e$, $\tilde{x}$, $\alpha_{ij}$ and $\beta_{ij}$ are bounded, and thus $x$ is bounded. This implies the boundedness of $\nabla f_i(x_i,t)$ and $\nabla_t f_i(x_i,t)$. Then, by \eqref{eq.e} and \eqref{eq.tildex}, $\dot{e}$ and $\dot{\tilde{x}}$ are bounded. By Barbalat's Lemma \cite{Nonlinear02khalil}, we have $\lim_{t\rightarrow \infty}e(t)=0$ and $\lim_{t\rightarrow \infty}\tilde{x}(t)=0$, which further implies $\lim_{t\rightarrow \infty}{x}_i(t)=r^*(t)$.

\balance
\bibliographystyle{IEEEtran}
\bibliography{IEEEabrv,bib_abb}

\begin{thebibliography}{10}
\providecommand{\url}[1]{#1}
\csname url@samestyle\endcsname
\providecommand{\newblock}{\relax}
\providecommand{\bibinfo}[2]{#2}
\providecommand{\BIBentrySTDinterwordspacing}{\spaceskip=0pt\relax}
\providecommand{\BIBentryALTinterwordstretchfactor}{4}
\providecommand{\BIBentryALTinterwordspacing}{\spaceskip=\fontdimen2\font plus
\BIBentryALTinterwordstretchfactor\fontdimen3\font minus
  \fontdimen4\font\relax}
\providecommand{\BIBforeignlanguage}[2]{{%
\expandafter\ifx\csname l@#1\endcsname\relax
\typeout{** WARNING: IEEEtran.bst: No hyphenation pattern has been}%
\typeout{** loaded for the language `#1'. Using the pattern for}%
\typeout{** the default language instead.}%
\else
\language=\csname l@#1\endcsname
\fi
#2}}
\providecommand{\BIBdecl}{\relax}
\BIBdecl

\bibitem{Nedic09}
A.~{Nedic} and A.~{Ozdaglar}, ``Distributed subgradient methods for multi-agent
  optimization,'' \emph{IEEE Trans. Autom. Control}, vol.~54, no.~1, pp.
  48--61, 2009.

\bibitem{wang2010control}
J.~Wang and N.~Elia, ``Control approach to distributed optimization,'' in
  \emph{2010 48th Annual Allerton Conference on Communication, Control, and
  Computing (Allerton)}.\hskip 1em plus 0.5em minus 0.4em\relax IEEE, 2010, pp.
  557--561.

\bibitem{gharesifard2013distributed}
B.~Gharesifard and J.~Cort{\'e}s, ``Distributed continuous-time convex
  optimization on weight-balanced digraphs,'' \emph{IEEE Trans. Autom.
  Control}, vol.~59, no.~3, pp. 781--786, 2013.

\bibitem{Rabbat04}
M.~Rabbat and R.~Nowak, ``Distributed optimization in sensor networks,'' in
  \emph{Third International Symposium on Information Processing in Sensor
  Networks}, 2004, pp. 20--27.

\bibitem{bhattacharya2011distributed}
S.~Bhattacharya and V.~Kumar, ``Distributed optimization with pairwise
  constraints and its application to multi-robot path planning,''
  \emph{Robotics: Science and Systems VI}, vol. 177, Jun. 2011.

\bibitem{QIN2022110113}
Z.~Qin, L.~Jiang, T.~Liu, and Z.-P. Jiang, ``Distributed optimization for
  uncertain {Euler–Lagrange} systems with local and relative measurements,''
  \emph{Automatica}, vol. 139, p. 110113, 2022.

\bibitem{Mahmoud16}
K.~Mahmoud, N.~Yorino, and A.~Ahmed, ``Optimal distributed generation
  allocation in distribution systems for loss minimization,'' \emph{IEEE Trans.
  Power Syst.}, vol.~31, no.~2, pp. 960--969, 2016.

\bibitem{huang2020distributionally}
W.~Huang, W.~Zheng, and D.~J. Hill, ``Distributionally robust optimal power
  flow in multi-microgrids with decomposition and guaranteed convergence,''
  \emph{IEEE Trans. Smart Grid}, vol.~12, no.~1, pp. 43--55, Jan. 2021.

\bibitem{su2009traffic}
W.~Su, ``Traffic engineering and time-varying convex optimization,'' Ph.D.
  dissertation, Dept. Elect. Eng., The Pennsylvania State Univ., State College,
  PA, USA, 2009.

\bibitem{Simonetto20}
A.~Simonetto, E.~Dall'Anese, S.~Paternain, G.~Leus, and G.~B. Giannakis,
  ``Time-varying convex optimization: Time-structured algorithms and
  applications,'' \emph{Proc. IEEE}, vol. 108, no.~11, pp. 2032--2048, 2020.

\bibitem{LI2023100904}
X.~Li, L.~Xie, and N.~Li, ``A survey on distributed online optimization and
  online games,'' \emph{Annu. Rev. Control}, vol.~56, p. 100904, 2023.

\bibitem{HATANAKA2016210}
T.~Hatanaka, R.~Funada, G.~Gezer, and M.~Fujita, ``Distributed visual {3-D}
  localization of a human using pedestrian detection algorithm: A
  passivity-based approach,'' \emph{IFAC-PapersOnLine}, vol.~49, no.~22, pp.
  210--215, 2016.

\bibitem{Shorinwa24}
O.~Shorinwa, T.~Halsted, J.~Yu, and M.~Schwager, ``Distributed optimization
  methods for multi-robot systems: Part 1—a tutorial,'' \emph{IEEE Robot.
  Autom. Mag.}, pp. 2--19, 2024.

\bibitem{Zhou22}
X.~Zhou, X.~Wen, Z.~Wang, Y.~Gao, H.~Li, Q.~Wang, T.~Yang, H.~Lu, Y.~Cao,
  C.~Xu, and F.~Gao, ``Swarm of micro flying robots in the wild,'' \emph{Sci.
  Robot.}, vol.~7, no.~66, p. eabm5954, 2022.

\bibitem{Intelligent2017}
H.~Menouar, I.~Guvenc, K.~Akkaya, A.~S. Uluagac, A.~Kadri, and A.~Tuncer,
  ``{UAV}-enabled intelligent transportation systems for the smart city:
  Applications and challenges,'' \emph{{IEEE} Commun. Mag.}, vol.~55, no.~3,
  pp. 22--28, Mar. 2017.

\bibitem{ZOU2020}
Y.~Zou, Z.~Meng, and Y.~Hong, ``Adaptive distributed optimization algorithms
  for {Euler–Lagrange} systems,'' \emph{Automatica}, vol. 119, p. 109060,
  Sep. 2020.

\bibitem{Simonetto16}
A.~Simonetto, A.~Mokhtari, A.~Koppel, G.~Leus, and A.~Ribeiro, ``A class of
  prediction-correction methods for time-varying convex optimization,''
  \emph{IEEE Trans. Signal Process.}, vol.~64, no.~17, pp. 4576--4591, 2016.

\bibitem{Fazlyab18}
M.~Fazlyab, S.~Paternain, V.~M. Preciado, and A.~Ribeiro,
  ``Prediction-correction interior-point method for time-varying convex
  optimization,'' \emph{IEEE Trans. Autom. Control}, vol.~63, no.~7, pp.
  1973--1986, 2018.

\bibitem{Ren17}
S.~Rahili and W.~Ren, ``Distributed continuous-time convex optimization with
  time-varying cost functions,'' \emph{IEEE Trans. Autom. Control}, vol.~62,
  no.~4, pp. 1590--1605, 2017.

\bibitem{Chu22}
C.~Wu, H.~Fang, X.~Zeng, Q.~Yang, Y.~Wei, and J.~Chen, ``Distributed
  continuous-time algorithm for time-varying optimization with affine formation
  constraints,'' \emph{IEEE Trans. Autom. Control}, vol.~68, no.~4, pp.
  2615--2622, 2023.

\bibitem{Wang22}
B.~Wang, S.~Sun, and W.~Ren, ``Distributed time-varying quadratic optimal
  resource allocation subject to nonidentical time-varying hessians with
  application to multiquadrotor hose transportation,'' \emph{IEEE Trans. Syst.,
  Man, Cybern., Syst.}, vol.~52, no.~10, pp. 6109--6119, 2022.

\bibitem{Sun23}
S.~Sun, J.~Xu, and W.~Ren, ``Distributed continuous-time algorithms for
  time-varying constrained convex optimization,'' \emph{IEEE Trans. Autom.
  Control}, vol.~68, no.~7, pp. 3931--3946, 2023.

\bibitem{Ding22}
Z.~Ding, ``Distributed time-varying optimization—an output regulation
  approach,'' \emph{IEEE Trans. Cybern.}, 2022, {DOI}:10.1109/TCYB.2022.
  3219295.

\bibitem{Huang20}
B.~Huang, Y.~Zou, Z.~Meng, and W.~Ren, ``Distributed time-varying convex
  optimization for a class of nonlinear multiagent systems,'' \emph{IEEE Trans.
  Autom. Control}, vol.~65, no.~2, pp. 801--808, 2020.

\bibitem{Sun17}
C.~Sun, M.~Ye, and G.~Hu, ``Distributed time-varying quadratic optimization for
  multiple agents under undirected graphs,'' \emph{IEEE Trans. Autom. Control},
  vol.~62, no.~7, pp. 3687--3694, 2017.

\bibitem{Santilli24}
M.~Santilli, A.~Furchì, G.~Oliva, and A.~Gasparri, ``A finite-time protocol
  for distributed time-varying optimization over a graph,'' \emph{IEEE Trans.
  Control Netw. Syst.}, vol.~11, no.~1, pp. 53--64, 2024.

\bibitem{Zheng24}
Y.~Zheng, Q.~Liu, and J.~Wang, ``A specified-time convergent multiagent system
  for distributed optimization with a time-varying objective function,''
  \emph{IEEE Trans. Autom. Control}, vol.~69, no.~2, pp. 1257--1264, 2024.

\bibitem{Algebraic}
C.~Godsil and G.~F. Royle, \emph{Algebraic graph theory}.\hskip 1em plus 0.5em
  minus 0.4em\relax Springer Science \& Business Media, 2001, vol. 207.

\bibitem{matrix}
R.~A. Horn and C.~R. Johnson, \emph{Matrix Analysis}, 2nd~ed.\hskip 1em plus
  0.5em minus 0.4em\relax Cambridge University Press, 2013.

\bibitem{fei12}
F.~Chen, Y.~Cao, and W.~Ren, ``Distributed average tracking of multiple
  time-varying reference signals with bounded derivatives,'' \emph{IEEE Trans.
  Autom. Control}, vol.~57, no.~12, pp. 3169--3174, 2012.

\bibitem{Cortes08}
J.~Cortes, ``Discontinuous dynamical systems,'' \emph{IEEE Control Syst. Mag.},
  vol.~28, no.~3, pp. 36--73, 2008.

\bibitem{kia2015}
S.~S. Kia, J.~Cortés, and S.~Martínez, ``Distributed convex optimization via
  continuous-time coordinate alorithms with discrete-time communication,''
  \emph{Automatica}, vol.~55, pp. 254--264, 2015.

\bibitem{Jiang24}
\BIBentryALTinterwordspacing
L.~Jiang, Z.-G. Wu, and L.~Wang, ``Distributed continuous-time optimization
  with uncertain time-varying quadratic cost functions.'' [Online]. Available:
  \url{https://arxiv.org/abs/2310.13541}
\BIBentrySTDinterwordspacing

\bibitem{Bastianello24}
N.~Bastianello, R.~Carli, and S.~Zampieri, ``Internal model-based online
  optimization,'' \emph{IEEE Trans. Autom. Control}, vol.~69, no.~1, pp.
  689--696, 2024.

\bibitem{tse2005wireless}
D.~Tse and P.~Viswanath, \emph{Fundamentals of wireless communication}.\hskip
  1em plus 0.5em minus 0.4em\relax Cambridge university press, 2005.

\bibitem{Merwaday15}
A.~Merwaday and I.~Guvenc, ``{UAV} assisted heterogeneous networks for public
  safety communications,'' in \emph{2015 IEEE Wireless Communications and
  Networking Conference Workshops (WCNCW)}, 2015.

\bibitem{EmergencyNetworks}
N.~Zhao, W.~Lu, M.~Sheng, Y.~Chen, J.~Tang, F.~R. Yu, and K.-K. Wong,
  ``{UAV}-assisted emergency networks in disasters,'' \emph{{IEEE} Wireless
  Commun. Mag.}, vol.~26, no.~1, pp. 45--51, Feb. 2019.

\bibitem{Nonlinear02khalil}
H.~K. Khalil, \emph{Nonlinear Systems}, 3rd~ed.\hskip 1em plus 0.5em minus
  0.4em\relax Prentice Hall, 2002.

\bibitem{KKKbook}
M.~Krsti\'{c}, I.~Kanellakopoulos, and P.~V. Kokotovi\'{c}, \emph{Nonlinear and
  Adaptive Control Design}.\hskip 1em plus 0.5em minus 0.4em\relax NY: John
  Wiley \& Sons, 1995.

\bibitem{Olfati}
R.~Olfati-Saber and R.~Murray, ``Consensus problems in networks of agents with
  switching topology and time-delays,'' \emph{IEEE Trans. Autom. Control},
  vol.~49, no.~9, pp. 1520--1533, 2004.

\bibitem{Wang10}
L.~Wang and F.~Xiao, ``Finite-time consensus problems for networks of dynamic
  agents,'' \emph{IEEE Trans. Autom. Control}, vol.~55, no.~4, pp. 950--955,
  2010.

\end{thebibliography}

\end{document}